\documentclass[journal,twoside]{asmb}

\renewenvironment{proof}{{\bfseries Proof.}}{\qedsymbol}
\renewcommand{\qedsymbol}{$\blacksquare$}

\newcounter{mytempeqncnt}

\begin{document}

\title{Secrecy Performance Analysis of Integrated RF-UWOC IoT Networks Enabled by UAV and Underwater-RIS}

\author{
\IEEEauthorblockN{Abrar Bin Sarawar, A. S. M. Badrudduza,~\IEEEmembership{Member,~IEEE,} Md. Ibrahim,~\IEEEmembership{Member,~IEEE,} Imran Shafique Ansari,~\IEEEmembership{Senior Member,~IEEE,} Heejung Yu,~\IEEEmembership{Senior Member,~IEEE}
}}

\twocolumn[
\begin{@twocolumnfalse}
\maketitle
\begin{abstract}
\section*{Abstract}

In the sixth-generation (6G) Internet of Things (IoT) networks, the use of UAV-mounted base stations and reconfigurable intelligent surfaces (RIS) has been considered to enhance coverage, flexibility, and security in non-terrestrial networks (NTNs). In addition to aerial networks enabled by NTN technologies, the integration of underwater networks with 6G IoT can be considered one of the most innovative challenges in future IoT. Along with such trends in IoT, this study investigates the secrecy performance of IoT networks that integrate radio frequency (RF) UAV-based NTNs and underwater optical wireless communication (UOWC) links with an RIS. Considering three potential eavesdropping scenarios (RF signal, UOWC signal, and both), we derive closed-form expressions for secrecy performance metrics, including average secrecy capacity, secrecy outage probability, probability of strictly positive secrecy capacity, and effective secrecy throughput. Extensive numerical analyses and Monte Carlo simulations elucidate the impact of system parameters such as fading severity, the number of RIS reflecting elements, underwater turbulence, pointing errors, and detection techniques on system security. The findings offer comprehensive design guidelines for developing such a network aiming to enhance secrecy performance and ensure secure communication in diverse and challenging environments.
\end{abstract}

\begin{IEEEkeywords}
\section*{Keywords} 

Reconfigurable intelligent surface, unmanned aerial vehicle, non-terrestrial network, underwater optical wireless communication, physical layer security

\end{IEEEkeywords}
\end{@twocolumnfalse}
]


\section{Introduction}

\subsection{Background}

Advancements in machine-type communications, particularly in Internet of Things (IoT), have led to the proliferation of connecting devices with diverse communication needs which poses significant challenges in network design. In non-terrestrial networks (NTNs), mobile base stations (BSs), mounted on unmanned aerial vehicles (UAVs) (a.w.a. drones), offer a flexible network infrastructure for uncovered areas by ground BSs \cite{ParkCCNC23}. Bypassing the need for extensive terrestrial infrastructure deployment, UAVs enhance network coverage and performance across various IoT applications and use cases. UAV-mounted BSs offer network flexibility and scalability for future networks by flying on-demand to locations, bypassing terrestrial constraints, and ensuring connectivity with both ground users and BSs \cite{AI6G_ICTE2022}. In addition to aerial networks, the integration of underwater networks into 6G IoT networks has been considered an essential research area, especially in IoT applications such as underwater monitoring with unmanned underwater vehicles (UUVs).
This makes them essential components of an integrated network spanning space, air, land, and sea, and optimal for serving predictable traffic demands of massive machine-type communication (MMTC) and IoT networks, where nodes are typically static and traffic demands are predictable \cite{mignardi2020unmanned}.

Nowadays, the reconfigurable intelligent surface (RIS) plays a pivotal role in the landscape of IoT-enabled 6G wireless applications, adeptly addressing challenges by modifying phase shifts and reflecting signals to intended locations  \cite{10413214}. An RIS is now a promising solution to address localized coverage gaps in crowded city instances and difficult residential propagation environments and can easily be mounted on indoor or outdoor walls. Beyond traditional NTN and/or underwater optical wireless communication (UOWC) networks, the integration of UAV-based NTN and  UOWC technologies demonstrates valuable impacts in achieving wider coverage and higher data rate \cite{10460296}. The integration of RISs into an NTN-UOWC dual-hop link has become essential, offering cutting-edge research avenues for seamlessly connecting terrestrial and underwater networks. Furthermore, the signal modification technique of an RIS can be leveraged to minimize propagated signals to eavesdroppers, introducing a new dimension to physical layer security (PLS) \cite{DCN_2024:Bae}.

\subsection{Literature Review}
RISs comprised of cost-effective meta-surfaces have been recognized as a prospective technology in wireless communication, offering creative ways to enhance performance across diverse dimensions. In relay systems, radio signals are amplified to extend transmission range. On the other hand, an RIS acts passively to reflect signals by utilizing the meta-surfaces as reflective and refractive radio mirrors and an RIS is regarded as a highly promising technology for optimizing data transmission efficiency. 
A comparative analysis of coverage extension of RIS-aided systems compared to direct link was unveiled in \cite{yang2020coverage}. In \cite{trigui2022performance} presented a reliable analytical basis for evaluating the maximum diversity order and outage probability (OP) in RIS-aided systems. When compared to other active RIS designs, the recommended design in \cite{tasci2022new} was proficient in both affordability and energy efficiency. For RIS-aided millimeter wave (mmWave) and sub-THz systems, channel estimation techniques were studied in \cite{CE_RIS_VTM22}. In \cite{TCOM24:RIS}, the design approaches of RIS phase shift were investigated. The concept of underwater RISs was introduced in \cite{underwater_RIS}. However, the detailed performance analysis of underwater links including RISs has been rarely investigated.  

Nowadays, dual-hop networks, especially those integrating RF and optical communication technologies, have the ability to broaden coverage and enable higher data rates. Recently, there has been growing interest in using UAVs to enhance the connectivity and coverage of dual-hop networks, especially in areas where traditional infrastructure is unavailable or impractical. In particular,  6G NTN enabled by UAVs as well as satellites have gathered research interests to realize three-dimensional (3D) global coverage. A notable UAV relay method was approached for a hybrid RF-FSO network in \cite{li2023ris} emphasizing the critical role of minimizing pointing errors to reduce OP. 
In \cite{xu2021performance}, ensuring the collaboration with IoT, the UAV-aided RF-FSO model was considered in terms of OP and bit error rate (BER). The authors demonstrated the closed-form expressions in the presence of atmospheric turbulence. However, for ensuring line-of-sight (LOS) communication, a UAV-aided triple-hop RF-FSO-UOWC network was demonstrated in \cite{agheli2021uav} whereby the end-to-end OP and average BER were evaluated. 

In wireless communication, ensuring the confidentiality of transmitted data is paramount for safeguarding sensitive information, and guaranteeing the success of strategic operations to uphold national security,  especially in military defense applications. A series of studies on the intricacies of secrecy evaluation shed light on innovative techniques e.g., introducing artificial noise (AN), using RISs \cite{ruku2023effects}, incorporating UAVs, etc., to enhance the security of communication systems \cite{Waqas_SENSOR21, Yu_TIFS18, RIS-aid_PLS_ACCESS21, RIS_HARQ, AN_power_PLS, AN_TD_tradeoff, YU2020611, DCN_2024:Son}. Beginning with \cite{xing2015secrecy}, the authors explored PLS mechanisms in simultaneous power and data transfer scenarios by introducing AN-aided transmission systems. In \cite{badrudduza2021security}, a secrecy analysis of RF-UOWC networks was conducted, revealing that a decrease in fading and underwater turbulence (UWT) could fortify the security of such communication systems. In parallel efforts, the authors of \cite{lou2021secrecy, ibrahim2021enhancing} focused on the secrecy performance of similar types of networks, employing sophisticated statistical distributions to model RF and UOWC links along with antenna selection approaches. The authors of \cite{lei2017secrecy} conducted a comprehensive investigation into the secrecy outage probability (SOP) of hybrid RF-FSO models, providing valuable insights into system behavior at high SNR regimes. The nuances of RF and FSO wiretapping were analyzed in \cite{erdogan2021secrecy}, uncovering the superior performance of FSO eavesdropping in weak turbulence conditions.  
Furthermore, the authors in \cite{hossain2022physical} had provided an RIS-aided hybrid RF-UOWC network, emphasizing the critical role of reflective elements in ensuring secure communication for users and thwarting potential eavesdroppers. Besides, researchers in \cite{lou2022physical} looked into how temperature gradients and air bubble levels affected the UOWC channel within UAV-based RF (NTN)-UOWC networks. 


\subsection{Motivation and Contribution}

It is noteworthy that the existing research predominantly explores secrecy performance in scenarios where the source remains stationary \cite{lei2017secrecy} limiting the applicability in dynamic environments. This gap in literature underscores the need to investigate new models that accommodate dynamic sources, particularly UAVs and low earth orbit (LEO) satellites in NTNs. Moreover, while some recent studies have considered secrecy frameworks within the hybrid RF-FSO paradigm, these investigations have primarily focused on terrestrial network (TN) scenarios, neglecting the unique challenges posed by UOWC. Additionally, although an RIS has shown promise in enhancing security in RF and FSO communications \cite{rahman2023ris} its potential in UOWC scenarios remains largely untapped. Therefore, there exists a critical research gap in exploring the integration of an RIS for enhancing both performance and security in UOWC environments. Furthermore, while previous studies have examined eavesdropping scenarios in either RF or UOWC paths individually \cite{badrudduza2021security, hossain2022physical}, simultaneous eavesdropping in both domains remains unexplored. This limitation hampers the development of comprehensive security strategies for hybrid RF-UOWC networks, necessitating novel approaches to address potential threats effectively. In light of these considerations, we address the research gap by introducing innovative models that incorporate UAVs as dynamic sources and leverage an RIS in the UOWC link for enhanced performance and security.  Taking into account the small-scale fluctuations of the RF links, the generalized $\kappa-\mu$ distribution is considered, which yields the widely-used Rice and Nakagami-$m$ distributions as special cases. All the underwater links are followed by a mixture Exponential-Generalized Gamma model (mEGG) provides a vivid representation of the random nature of irradiance fluctuations on the laser beams. In the end, the following outline holds the main contribution of this study:

\begin{itemize}

\item Our primary contribution lies in addressing the gap in research by proposing a UAV-based RF NTN-UOWC model that encompasses three potential eavesdropping scenarios: RF, UOWC, and simultaneous RF-UOWC eavesdropping. Unlike previous studies that primarily focused on stationary sources, 
our model accounts for the dynamic nature of UAVs, enabling a more accurate representation of real-world scenarios. Furthermore, apart from the existing secrecy analysis on hybrid RF NTN-UOWC networks, 
we extend our analysis to include the deployment of a RIS in UOWC environments offering a more comprehensive security analysis.

\item We derive the cumulative distribution function (CDF) for the dual-hop RF-UOWC link, followed by the derivation of analytical expressions for secrecy metrics such as average secrecy capacity (ASC), SOP, probability of strictly positive secrecy capacity (SPSC), and effective secrecy throughput (EST), validated through Monte Carlo (MC) simulations. We also provide asymptotic expressions of SOP that facilitate a deeper understanding of the system's behavior at a high SNR regime. To the best of the authors' knowledge, the derived expressions offer a level of generality due to the incorporation of statistical distributions that are themselves generalized, 
and have not been previously reported in the literature, representing a novel contribution to the field.

\item Through the derived closed-form expressions, we investigate the influence of fading parameters, number of reflecting elements, UWT, pointing error, detection techniques, etc., on the overall behavior of the system. Additionally, we offer a graphical comparison between RIS-aided and non-RIS systems, providing valuable insights into the efficacy of an RIS deployment in UOWC environments. 
    
\end{itemize}

\subsection{Organization}
The paper is organized into several key sections to systematically address different aspects of the proposed system. Section II provides a pictorial representation of the entire system. In Section III, we detail the channel models utilized in both RF NTN and UOWC links. Performance evaluation metrics are thoroughly examined in Sections IV through VII. Section VIII presents the results of our analytical analyses and simulation studies, offering insights into the system's behavior and performance under various conditions. Finally, Section IX outlines the conclusions drawn from our findings.

\section{System Model}
\begin{figure}[t]
\centerline{\includegraphics[width=0.5\textwidth,angle=0]{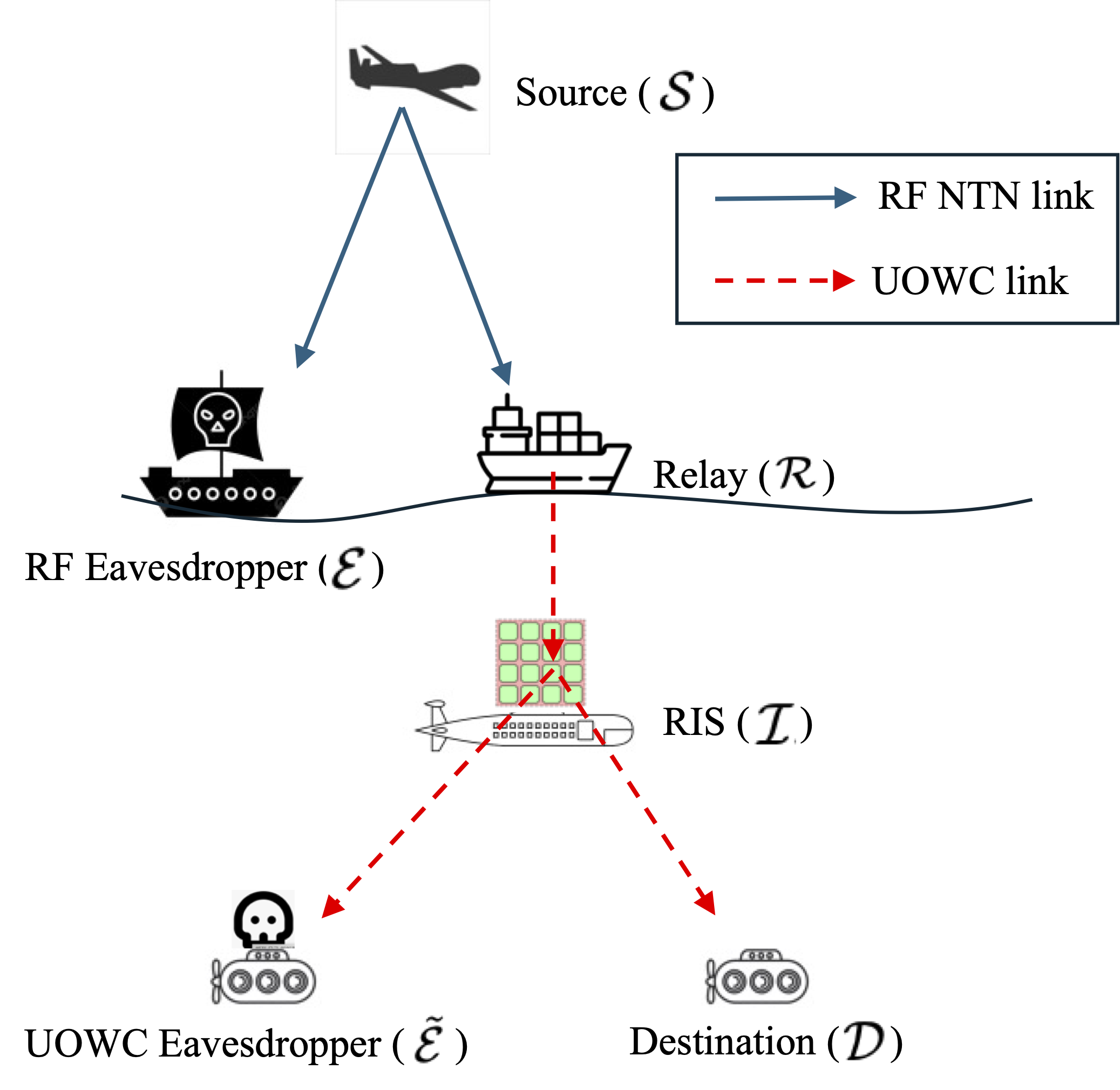}}
    \caption{System model incorporating a source ($\mathcal{S}$), a relay ($\mathcal{R}$), an RIS ($\mathcal{I}$), a destination ($\mathcal{D}$) and two eavesdroppers ($\mathcal{E}$ and $\mathcal{\tilde{E}}$).}
    \label{fig:fig1e}
\end{figure}
The suggested system model illustrates a dual-hop UAV-based RF NTN and an RIS-assisted UOWC wireless communication setup. The initial stage involves a UAV as the source $\mathcal{S}$, a relay $\mathcal{R}$, and an eavesdropper $\mathcal{E}$ in the first hop. The subsequent phase includes an RIS $\mathcal{I}$, a destination $\mathcal{D}$, and an additional eavesdropper $\mathcal{\tilde{E}}$ situated below the water's surface in the second hop. Such a type of UAV-based NTN communication scenario is particularly beneficial in remote areas with diverse IoT applications, where UAV-based BSs can efficiently serve a multitude of devices across different zones during a single flight, maximizing network coverage and performance. In this system, a UAV hovering in the air transmits information during two distinct time slots. During the initial slot, it transfers data to $\mathcal{R}$ positioned on a ship. Conversely, in the subsequent slot, $\mathcal{R}$ receives the RF signal, transforms it into its optical counterpart, and then sends it to $\mathcal{\tilde{E}}$ via $\mathcal{I}$. Here, to maintain uninterrupted underwater communication, $\mathcal{I}$ is deployed on an underwater drone. Unwanted entities, identified as eavesdroppers, can intercept private data through the $\mathcal{S}$-$\mathcal{E}$ and $\mathcal{R}$-$\mathcal{I}$-$\mathcal{\tilde{E}}$ links and monitoring these eavesdropping activities, the three scenarios are considered as follows.

\begin{itemize}
    \item In Scenario-I, the RF NTN communication link between $\mathcal{S}$ and $\mathcal{R}$ is intercepted by $\mathcal{E}$. Consequently, both $\mathcal{R}$ and $\mathcal{\tilde{E}}$ intercept and receive the exact signals transmitted from $\mathcal{S}$.
    
    \item In Scenario-II, both $\mathcal{D}$ and $\mathcal{\tilde{E}}$ capture the identical UOWC signal transmitted from $\mathcal{R}$ through $\mathcal{I}$.  
    
    \item Scenario-III incorporates the consequence of conjointly eavesdropping action at both the RF NTN and UOWC links. 
\end{itemize}

However, within the proposed framework, $\mathcal{R}$ operates as a transceiver with $L_{\mathcal{R}}$ receive antennas and a transmit aperture, $\mathcal{S}$ is equipped with a single antenna, whereas $\mathcal{E}$ is equipped with $L_{\mathcal{E}}$ receive antennas. Moreover, $\mathcal{N}$ denotes the count of reflective meta-surfaces within $\mathcal{I}$, wherein $\mathcal{D}$ and $\mathcal{\tilde{E}}$ solely possess a single photodetector for receiving optical signals. All the RF links adhere to the $\kappa-\mu$ fading model, while the UOWC network, encompassing the $\mathcal{R}$-$\mathcal{I}$-$\mathcal{\tilde{E}}$ and $\mathcal{R}$-$\mathcal{I}$-$\mathcal{D}$ links, encounters the mEGG distribution coupled with pointing error impairments.

\section{Problem Formulation}

\subsection{SNRs of UAV-aided RF NTN Links}

The received signal at $l$th antenna ($l=1,2,\ldots,L_{j}$) of $\mathcal{R}$ and $\mathcal{E}$ can be expressed as
\begin{align}
 \label{eq:eq1}
    y_{jl}=\sqrt{\frac{\varrho _{jl}\mathcal{P}_{T}}{q_{j}^{\alpha_{jl} }}} h_{jl}x_{s}+z_{jl},
\end{align}
where $j\in \{ \mathcal{R}, \mathcal{E} \}$, $h_{jl}$ is the channel gain, $\varrho_{jl}=G_{s}G_{q_{jl}}\left ( \frac{\Upsilon _{s}}{4\pi} \right )^{2}$, $G_{s}$ is the gain of the antenna at $\mathcal{S}$, $G_{q_{jl}}$ is the gain of each antenna at $\mathcal{R}$ and $\mathcal{E}$, $\Upsilon_{s}$ outlines the wavelength, the symbol for the path loss components are $\alpha_{jl}$, the transmitted power is denoted by $\mathcal{P}_{T}$ wherein $x_{s}$ is expressed as the transmitted signal and complex Gaussian noise with zero mean value and fixed
standard deviation $\sigma _{z_{jl}}^{2}$ is symbolized by $z_{jl}$. According to (\ref{eq:eq1}), the received instantaneous SNRs can be accomplished as follows. 
\begin{align}
\label{eq:eq2}
    \gamma_{jl}=\varrho _{jl} \left |h_{jl} \right |^{2}q_{j}^{-\alpha_{jl} }\bar{\gamma} _{jl},
\end{align} 
where the average SNRs are expressed by $\bar\gamma_{jl}=\frac{\mathcal{P}_{T}}{\sigma_{z_{jl}} ^{2}}$. Is it noteworthy that due to the deployment of $L_{j}$-branch MRC receivers at $\mathcal{R}$ and $\mathcal{E}$, the output SNRs at the respective combiners can be defined as $\gamma_{j}=\sum_{l=0}^{L_{j}}\gamma_{jl}$.

\subsection{SNRs of RIS-aided UOWC Links}
In the second hop, the channel coefficient of the $\mathcal{R}$-$\mathcal{I}$ link is specified by $h_{u}$ $(u=1,2,\ldots,\mathcal{N})
$ and in the similar manner, $g_{u,n}$ indicates the channel gain of the $\mathcal{I}$-$\mathcal{D}$ and $\mathcal{I}$-$\mathcal{\tilde{E}}$ links. Thus expressions of the received signal at $\mathcal{D}$ and $\mathcal{\tilde{E}}$ is given by
\begin{align}
\label{aaa}
    y_{n}=\left [\sum_{u=1}^{\mathcal{N}}h_{u}e^{j\theta_{u,n}}g_{u,n}\right]y_{R}+w_{n},
\end{align} 
where $n\in \{ \mathcal{D}, \mathcal{\tilde{E}} \}$ and $y_{R}$ is the re-transmitted signal from $\mathcal{R}$. Here, $h_{u}=\alpha_{u}e^{j\varphi _{u}}$, $g_{u,n}=\beta_{u,n}e^{j\varrho_{u,n}}$, $\alpha_{u}$ and $\beta _{u,n}$ are mEGG random variables (RVs), $\varphi _{u}$ and $\varrho_{u,n}$ outline the uniformly distributed phases of the corresponding links with ranges $[0,2\pi)$, $\theta _{u,n}$ stands for the phase shift induced by $u$th reflecting element of $\mathcal{I}$ and $w_{n}\sim\widetilde{\mathcal{N}}(0, N_{n})$ is the AWGN samples with a power of $N_{n}$.
Equation \eqref{aaa} can be expressed in an alternative form as
\begin{align}
y_{n}=\mathbf{g}_{n}^{T} \theta_{n} \mathbf{h} y_{R}+w_{n},
\end{align}
where $\mathbf{h}=[h_{1}h_{2}\ldots h_{\mathcal{N}}]^{T}$ and $\mathbf{g}_{n}=[g_{1,n}g_{2,n}\ldots g_{\mathcal{N},n}]^{T}$ are the channel coefficient vectors and the diagonal matrices for phase shifts are given by $\theta_{n} = \text{diag}([e^{j\theta _{1,n}}e^{j\theta _{2,n}}\ldots e^{j\theta_{\mathcal{N},n}}])$ which involves the phase shifts carried out by the RIS elements. The instantaneous SNRs at $\mathcal{D}$ and $\mathcal{\tilde{E}}$ can now be presented as
\begin{align}
\gamma_{n}=\left [\sum_{u=1}^{\mathcal{N}} \alpha _{u} \beta _{u,n}e^{j(\theta_{u,n}-\varphi_{u}-\varrho _{u,n})}\right ]^{2}\bar{\gamma}_{n},
\end{align} 
where $\bar\gamma_{n}=\frac{\mathcal{P}_{R}}{N_{n}}$ indicate the average SNR of $\mathcal{R}$-$\mathcal{I}$-$\mathcal{D}$ and $\mathcal{R}$-$\mathcal{I}$-$\mathcal{\tilde{E}}$ links. It is true that the RIS elements serve as intelligent reflecting surfaces, manipulating wireless signals to optimize reception quality while thwarting eavesdropping attempts by adjusting reflection coefficients, beamforming, power allocation, and phase shifts. Achievement of this dual optimization enhances both secure communication and overall system performance against malicious interception. Similar to \cite{chen2021impact, hossain2022physical}, considering the worst-case scenario (i.e., maximum eavesdropper's SNR), the optimal selection of $\theta _{u_{n}}$ that maximizes the instantaneous SNR is $\theta _{u_{n}}=\varphi _{u}+\varrho _{u_{n}}$. Hence, the maximized SNR can be found at $\mathcal{D}$ and $\mathcal{\tilde{E}}$ is
\begin{align}
\gamma_{n}=\left ( \sum_{u=1}^{\mathcal{N}}\alpha_{u}\beta_{u,n}   \right )^{2}\bar\gamma _{n}.
\end{align}
 
\subsection{SNR of Dual-hop Link}
The received SNR at $\mathcal{D}$ employing variable gain relay, is given by
\begin{subequations}
\begin{align}
\label{lah}
\gamma_{eq}&=\frac{\gamma_{R}\gamma_{\mathcal{D}}}{\gamma_{R}+\gamma_{\mathcal{D}}+1}
\\
\label{ldy}
&\cong \text{min} \left  \{ \gamma_{R},\gamma_{\mathcal{D}} \right \}.
\end{align}
\end{subequations}
\begin{remark}
The analytical modeling of the proposed system, incorporating SNR statistics as described in \eqref{lah}, is notably challenging due to mathematical intractability. Consequently, we resort to the approximation provided in \eqref{ldy}, which remains valid even in the context of DF relay, as testified in vast existing literature, such as \cite{ahmed2023enhancing}.
\end{remark}

\subsection{Statistics of UAV-aided RF NTN Links}

Since the UAV is hovering in the air, the distance between $\mathcal{S}$ to $\mathcal{R}$ and $\mathcal{S}$ to $\mathcal{E}$ can be modeled as a 3D random waypoint model (RWP) and the corresponding random distances are denoted by $q_{j}$. Following this 3D model, the {probability
distribution function (PDF) of $q_{j}$ is given by
\begin{align}
f_{q_{j}}(q)=\sum_{i_{j}=1}^{m}\frac{C_{i_{j}}}{D_{j}^{\beta _{i_{j}}+1}}q^{\beta_{i_{j}}}, \quad 0\leq q\leq D_{j}  
\end{align}
where $m$ = 3, $C_{i_{j}}=\frac{1}{72}[735,-1190,455]$ and $\beta _{i_{j}}= [2,4,6]$ \cite{1624340}. The $\mathcal{R}$ is presumed to be placed at the center of a sphere with a maximum radius of $D_{j}$ in the 3D network topology of the RWP model. 
\begin{remark}
It is noteworthy that the RWP is widely used for modeling UAV links due to its capability to capture the dynamic nature of UAV movement since the trajectories of UAVs are not predetermined. Furthermore, the RWP allows for the simulation of realistic UAV mobility patterns, which is essential for accurately assessing network dynamics in communication applications. Additionally, the parameters of RWP can be adjusted to mimic various UAV behaviors, proving it a versatile tool for UAV link modeling \cite{1624340}.   
\end{remark}

Since the $\mathcal{S}$-$\mathcal{R}$ and $\mathcal{S}$-$\mathcal{E}$ links experience $\kappa-\mu$ fading distribution, the PDF of $h_{jl}$ can be written as \cite{yacoub2007kappa}
\begin{align}
    f_{h_{jl}}(x)=A_{j}
    x^{\mu_{j}}e^{-B_{j}x^{2}}I_{\mu_{j} -1}[M_{j}x],
    \label{eq:eq5}
\end{align}
where $A_{j}=\frac{2\mu_{j} (1+\kappa_{j})^{\frac{\mu_{j} +1}{2}}}{\kappa_{j} ^{\frac{\mu_{j} -1}{2}}e^{\mu_{j} \kappa_{j}}}$, $B_{j}=\mu_{j} (1+\kappa_{j})$, $M_{j}=2\mu_{j} \sqrt{\kappa_{j} (1+\kappa_{j} )}$. Herein, $\kappa_{j}$ and $\mu_{j}$ are known as the fading parameters, and $I_{v}(.)$ is the first kind of modified Bessel function of order $v$ \cite{zwillinger2007table}. 

\begin{remark}
The $\kappa-\mu$ fading channel model serves as a valuable tool enabling accurate channel modeling and aiding in the optimization of communication systems for diverse and realistic wireless application scenarios, e.g., wireless sensor network, satellite communication, underwater acoustic networks, mmWave communication for 5G and beyond networks, and vehicular networks where signals encounter rapid variations due to vehicle mobility. Furthermore, based on the values of $\kappa$ and $\mu$, the $\kappa-\mu$ distribution exhibits enormous generality since it yields various multipath channels \cite{9615217} as special cases.
\end{remark}

\begin{lemma}
The PDF of $\gamma_{j}$ can be expressed as
\begin{align}
\label{eq:eq23}
f_{\gamma_{j}}(\gamma)=\sum_{i_{j}=1}^{m}\sum_{k_{j}=0}^{p}K_{1_{j}}\gamma^{-1}\text{G}_{1,2}^{1,1} \left[K_{2_{j}}\gamma \middle| 
\begin{matrix}1-\Psi _{i_{j}} \\
L_{j}\mu_{j}+k_{j},-\Psi _{i_{j}}
\end{matrix}\right],
\end{align}
where $K_{1_{j}}=\frac{1}{\alpha_{j}e^{L_{j}\mu_{j}\kappa_{j}}}C_{i_{j}}V_{j}(k_{j},p,L_{j}\mu _{j}-1)\left (L_{j}\mu_{j}\kappa_{j}  \right )^{k_{j}}$, $K_{2_{j}}=\frac{B_{j}D_{j}^{\alpha_{j}}}{\varrho_{j}L_{j}\bar{\gamma_{j}}}$, and $\Psi_{i_{j}}=\frac{1}{\alpha_{j}}\left ( \beta_{i_{j}}+1 \right )$.
\end{lemma}

\begin{proof}
See Appendix \ref{A1}.
\end{proof}

\begin{lemma}
The CDF of $\gamma_{j}$ can be expressed as
\begin{align}
F_{\gamma_{j}}(\gamma)=\sum_{i_{j}=1}^{m} \sum_{k_{j}=0}^{p}K_{1_{j}}\text{G}_{2,3}^{1,2} \left[K_{2_{j}}\gamma  \middle| 
\begin{matrix}1-\Psi _{i_{j}},1 \\
L_{j}\mu_{j}+k_{j},0,-\Psi _{i_{j}}
\end{matrix}\right].
\label{eq:eq24}
\end{align}
\end{lemma}

\begin{proof}
See Appendix \ref{A2}.
\end{proof}

\begin{lemma}
The asymptotic CDF of $\gamma_{j}$ can be expressed as
\begin{align}
\nonumber
F_{\gamma_{j}}^{\infty}(\gamma)&=\sum_{i_{j}=1}^{m}\sum_{k_{j}=0}^{p}\frac{K_{1_{j}}}{(K_{2_{j}}\gamma)^{-L_{j}\mu_{j}-k_{j}}}\\
&\times\frac{\prod_{l_{1}=1}^{2}\Gamma(T_{1,l_{1}}+L_{j}\mu_{j}+k_{j})}{\prod_{l_{1}=2}^{3}\Gamma(T_{2,l_{1}}+L_{j}\mu_{j}+k_{j})},
\label{eq:eq18n}
\end{align}
where $T_{1}=(\Psi _{i_{j}},0)$, $T_{2}=(1,1+\Psi _{i_{j}})$, and the $y^{th}$ term of $T_{x}$ can be represented as $T_{x,y}$.
\end{lemma}

\begin{proof}
See Appendix \ref{A3}.
\end{proof}

\begin{figure*}[!b]
\hrulefill
\setcounter{equation}{13}
\begin{subequations}
\begin{align}
\aleph_{1}&=\frac{(w_{I}w_{n}\xi_{I}^{2}\xi_{n}^{2})(J_{I}J_{n}\nu_{I}\nu_{n})^{p}\left \{ \Gamma (1+p) \right \}^{2}}{\left ( \xi _{n}^{2}+p \right )\left ( \xi _{I}^{2}+p \right )}.
\label{eqn:eq19n}
\\
\aleph_{2}&=\frac{(1-w_{n})w_{I}\xi_{I}^{2}\xi_{n} ^{2}(\nu_{I}b_{n}c_{n}J_{I}J_{n})^{c_{n}p}\Gamma (a_{n}+p)\Gamma (1-\Delta\left ( c_{n},0 \right )+p)\Gamma (1-\Delta\left ( c_{n},1-\xi_{I}^{2} \right )+p)}{c_{n}^{\frac{1}{2}}(2\pi)^{\frac{1}{2}(c_{n}-1)}\left ( \frac{\xi _{n}^{2}}{c_{n}}+p \right )\Gamma(1-\Delta \left (c_{n},-\xi_{I}^{2}\right)+p)}.
\label{eqn:eq20n}
\\
\aleph_{3}&=\frac{(1-w_{I})w_{n}\xi_{I} ^{2}\xi_{n} ^{2}(\nu_{n}b_{I}c_{I}J_{n}J_{I})^{c_{I}p}\Gamma (a_{I}+p)\Gamma (\Delta\left ( c_{I},1 \right )+p)\Gamma (\Delta\left ( c_{I},\xi_{n}^{2} \right )+p)}{c_{I}^{\frac{1}{2}}(2\pi)^{\frac{1}{2}(c_{I}-1)}\left ( \frac{\xi _{I}^{2}}{c_{I}}+p \right )\Gamma(a_{I})\Gamma (\Delta\left ( c_{I},1+\xi_{n}^{2} \right )+p)}.
\label{eqn:eq21n}
\\
\aleph_{4}&=\frac{(1-w_{I})(1-w_{n})\xi_{I}^{2}\xi_{n}^{2}(b_{I}J_{I})^{c_{I}p}(b_{n}J_{n})^{c_{n}p}\Gamma (a_{n}+p)\Gamma (a_{I}+p)}{\left ( \frac{\xi _{n}^{2}}{c_{n}}+p \right )\left ( \frac{\xi _{I}^{2}}{c_{I}}+p \right )\Gamma(a_{I})\Gamma(a_{n})}.
\label{eqn:eq22n}
\end{align}
\end{subequations}
\end{figure*}

\subsection{Statistics of RIS-aided UOWC Links}
The optical beam transmitted from $\mathcal{R}$ reaches $\mathcal{D}$ after reflection from $\mathcal{I}$. The UWT due to air bubbles,  temperature gradients, and water salinity (fresh and salty water) causes irradiance fluctuations which can be modeled effectively with mEGG distribution \cite{zedini2019unified}. Besides, as $\mathcal{I}$ is placed on an underwater drone in water, due to the waves in the water, $\mathcal{I}$ moves slightly in the water resulting in a pointing error that can be modeled using \cite{farid2007outage}. We assume $\alpha_{u}$ and $\beta_{u,n}$ are mEGG distributed with pointing error impairments. Hence, the resultant RV associated with an RIS-aided $\mathcal{R}-\mathcal{I}-\mathcal{D}$ link for $u$th element of $\mathcal{I}$ is obtained as $Y_{u,n}=\alpha_{u}\beta_{u,n}$. Due to $\mathcal{N}$ reflecting elements, we obtain `$\mathcal{N}$' number of RVs as $\chi_{n}=\sum_{u=1}^{\mathcal{N}}Y_{u,n}$, which can be easily approximated by the gamma distribution that results from expanding the first term of a Laguerre series. The corresponding shape and scale parameters are given by $\rho_{n}=\frac{(\mathbb{E}(\chi_{n}))^{2}}{\text{Var}(\chi_{n})}$, and $w_{n}=\frac{\text{Var}(\chi_{n})}{\mathbb{E}(\chi_{n})}$, respectively. Since the expectation operator is a linear one, $\mathbb{E}(\chi_{n})=\mathcal{N}\mathbb{E}(Y_{u,n})$, 
$\text{Var}(\chi_{n})=\mathcal{N}\text{Var}(Y_{u,n})$. Here, $p$th moment of $Y_{u,n}$ is given by
\setcounter{equation}{12}
\begin{align}
    \mathbb{E}(Y_{u,n}^{p})=\aleph_{1}+\aleph_{2}+\aleph_{3}+\aleph_{4},
    \label{eq:eq19n}
\end{align}
where $\aleph_{1}$, $\aleph_{2}$, $\aleph_{3}$, and $\aleph_{4}$ are given in \eqref{eqn:eq19n}, \eqref{eqn:eq20n}, \eqref{eqn:eq21n}, and \eqref{eqn:eq22n}, respectively. Finally, using the RV transform, the PDF of $\gamma_{n}$ is expressed as follows \cite{rakib2024ris}. 
\setcounter{equation}{14}
\begin{align}
f_{\gamma_{n}}(\gamma)=\Upsilon_{1_{n}}\gamma ^{\frac{\rho_{n}}{r_{n}}-1}\text{G}_{0,1}^{1,0} \left[\Upsilon_{2_{n}}\gamma^{\frac{1}{r_{n}}}  \middle| 
\begin{matrix} 
- 
\\
0
\end{matrix}\right],
\label{eq:eq25}
\end{align}
where 
$\Upsilon_{1_{n}}= \frac{(\mathbb{E}(Y_{u,n}))^{\rho_{n}}}{r_{n}\Gamma(\rho _{n})w_{n}^{\rho _{n}}\tau_{r_{n},n}^{\frac{\rho _{n}}{r_{n}}}}$, 
$\Upsilon_{2_{n}}=\frac{\mathbb{E}(Y_{u,n})}{w_{n}\tau_{r_{n},n}^{\frac{1}{r_{n}}}}$,  
$r_{n}$ designates the detection method (i.e. $r_{n}$ = 1 implies the heterodyne detection (HD) method and $r_{n}$ = 2 represents the intensity modulation/direct detection (IM/DD) method), $\tau_{r_{n},n}=\frac{(\eta \mathbb{E}(Y_{u,n}))^{r_{n}}}{N_{n}}$ symbolizes electrical SNRs, $\xi_{I}$ and $\xi_{n}$ denote the pointing errors and $J_{I}$ and $J_{n}$ indicate the pointing loss. For $\mathcal{R}-\mathcal{I}$ link, $\nu_{I}$ stands for the exponential distribution parameter, $0< w_{I}< 1$ indicates the mixture weight, and the GG distribution parameters are specified by $a_{I}$, $b_{I}$ and $c_{I}$, respectively whereas for $\mathcal{I}-\mathcal{D}$ and $\mathcal{I}-\mathcal{\tilde{E}}$ links, those parameters are defined as $\nu_{n}$, $w_{n}$, $a_{n}$, $b_{n}$ and $c_{n}$. The values of those parameters are considered from \cite{zedini2019unified} by taking into account different air bubble levels, and temperature gradients in fresh and salty water that induce diverse turbulence situations. In consequence, the CDF of $\gamma_{n}$ can be presented as
\begin{align}
       F_{\gamma_{n}}(\gamma)=\Upsilon_{3_{n}}\gamma ^{\frac{\rho_{n}}{r_{n}}}\text{G}_{1,r_{n}+1}^{r_{n},1} \left[\Upsilon_{4_{n}}\gamma \middle| \begin{matrix} 1-\frac{\rho_{n}}{r_{n}} \\
                                        0,\Upsilon_{5_{n}},-\frac{\rho_{n}}{r_{n}}
                          \end{matrix}\right],
                          \label{eq:eq27}
\end{align} where $\Upsilon_{3_{n}}=\frac{(\mathbb{E}(Y_{u,n}))^{\rho _{n}}}{\sqrt{r_{n}}\Gamma(\rho_{n})w_{n}^{\rho_{n}}\tau _{r_{n},n}^{\frac{\rho _{n}}{r_{n}}}(2\pi)^{\frac{r_{n}-1}{2}}}$, $\Upsilon_{4_{n}}=\frac{(\mathbb{E}(Y_{u,n}))^{r_{n}}}{r^{r_{n}}w_{n}^{r_{n}}\tau_{r_{n},n}}$ and $\Upsilon_{5_{n}}=\frac{r_{n}-1}{r_{n}}$. Similar to \eqref{eq:eq18n}, $F_{\gamma_{n}}(\gamma)$ can be asymptotically expressed as
\begin{align}
\nonumber
F^{\infty}_{\gamma_{n}}(\gamma)&=\Upsilon_{3_{n}}\gamma ^{\frac{\rho_{n}}{r_{n}}}\sum_{k_{2}=1}^{r_{n}}\left ( \frac{1}{\Upsilon_{4_{n}}\gamma} \right )^{T_{3,k_{2}}-1}
\\
&\times \frac{\prod_{l_{2}=1;l_{2}\neq k_{2}}^{r_{n}}\Gamma(T_{3,k_{2}}-T_{3,l_{2}})\Gamma(1+\frac{\rho_{n}}{r_{n}}-T_{3,k_{2}})}{\prod_{l_{2}=r_{n}+1}^{r_{n}+1}\Gamma(1+T_{3,l_{2}}-T_{3,k_{2}})},
\label{eq:eq26n}
\end{align}
where $T_{3}=(1,1-\Upsilon_{5_{n}},1+\frac{\rho_{n}}{r_{n}} )$.

\begin{remark}
The unified mEGG distribution emerges as the optimal probability distribution for characterizing fluctuations in underwater optical signal irradiance, attributed to different levels of air bubbles, temperature, and salinity gradients. This distribution not only demonstrates exceptional agreement with measured data across varying turbulence conditions but also offers a more manageable model, facilitating the derivation of straightforward expressions for various performance metrics in UOWC systems. Furthermore, in the case of uniform temperature, the intensity of the incident laser beam can be described
by the exponential-Gamma (EG) distribution which is a special case of the mEGG model \cite{zedini2019unified}.
\end{remark}

\subsection{Statistics of Dual-hop Link}
The expression for the CDF of end-to-end SNR $\gamma_{eq}$ of dual-hop RF NTN-UOWC system is given by
\begin{align}
    F_{\gamma_{eq}}(\gamma )=F_{\gamma_{R}}(\gamma )+F_{\gamma_{\mathcal{D}}}(\gamma)-F_{\gamma_{R}}(\gamma )F_{\gamma_{\mathcal{D}}}(\gamma).
    \label{eq:eq28}
\end{align} 
Replacing (\ref{eq:eq24}) and (\ref{eq:eq27}) into (\ref{eq:eq28}), the CDF of $\gamma_{eq}$ is derived as
\begin{align}
\nonumber
F_{\gamma_{eq}}(\gamma)&=\Upsilon_{3_{\mathcal{D}}} \gamma ^{\frac{\rho_{\mathcal{D}}}{r_{\mathcal{D}}}}
\text{G}_{1,r_{D}+1}^{r_{D},1} \left[\Upsilon_{4_{\mathcal{D}}}\gamma \middle| \begin{matrix} 1-\frac{\rho_{\mathcal{D}}}{r_{\mathcal{D}}} 
\\
0,\Upsilon_{5_{\mathcal{D}}},-\frac{\rho_{\mathcal{D}}}{r_{\mathcal{D}}}
\end{matrix}\right]
\\\nonumber
&+\sum_{i_{R}=1}^{m}\sum_{k_{R}=0}^{p}K_{1_{R}}\text{G}_{2,3}^{1,2} \left[K_{2_{R}}\gamma  \middle| \begin{matrix}1-\Psi _{i_{R}},1 
\\
L_{R}\mu_{R}+k_{R},0,-\Psi _{i_{R}}
\end{matrix}\right]
\\
&\times \biggl(1- \Upsilon_{3_{\mathcal{D}}}\gamma ^{\frac{\rho_{\mathcal{D}}}{r_{\mathcal{D}}}}
G_{1,r_{\mathcal{D}}+1}^{r_{\mathcal{D}},1} \left[\Upsilon_{4_{\mathcal{D}}}\gamma \middle| \begin{matrix} 1-\frac{\rho_{\mathcal{D}}}{r_{\mathcal{D}}} 
\\
0,\Upsilon_{5_{\mathcal{D}}},-\frac{\rho_{\mathcal{D}}}{r_{\mathcal{D}}}
\end{matrix}\right]\biggl).
\label{eq:eq29}
\end{align} 
$F_{\gamma_{eq}}(\gamma)$ can be written asymptotically as
\begin{align}
    F^{\infty}_{\gamma_{eq}}(\gamma)\approx F^{\infty}_{\gamma_{R}}(\gamma)+F^{\infty}_{\gamma_{\mathcal{D}}}(\gamma).
    \label{eq:eq48e}
\end{align}
Substituting \eqref{eq:eq18n} and \eqref{eq:eq26n} into \eqref{eq:eq48e}, the asymptotic equivalent CDF of the dual-hop link can be obtained as 
\begin{align}
\nonumber
    F^{\infty}_{\gamma_{eq}}(\gamma)&=\sum_{i_{R}=1}^{m}\sum_{k_{R}=0}^{p}\frac{K_{1_{R}}}{(K_{2_{R}}\gamma)^{-L_{R}\mu_{R}-k_{R}}}\\
    \nonumber
    &\times\frac{\prod_{l_{1}=1}^{2}\Gamma(T_{1,l_{1}}+L_{R}\mu_{R}+k_{R})}{\prod_{l_{1}=2}^{3}\Gamma(T_{2,l_{1}}+L_{R}\mu_{R}+k_{R})}
\\\nonumber
&+\Upsilon_{3_{\mathcal{D}}}\gamma ^{\frac{\rho_{\mathcal{D}}}{r_{\mathcal{D}}}}\sum_{k_{2}=1}^{r_{\mathcal{D}}}\left ( \frac{1}{\Upsilon_{4_{\mathcal{D}}}\gamma} \right )^{T_{3,k_{2}}-1}
\\
&\times \frac{\prod_{l_{2}=1;l_{2}\neq k_{2}}^{r_{\mathcal{D}}}\Gamma(T_{3,k_{2}}-T_{3,l_{2}})\Gamma(1+\frac{\rho_{\mathcal{D}}}{r_{\mathcal{D}}}-T_{3,k_{2}})}{\prod_{l_{2}=r_{\mathcal{D}}+1}^{r_{\mathcal{D}}+1}\Gamma(1+T_{3,l_{2}}-T_{3,k_{2}})}.
\label{eq:eq49e}
\end{align}

\section{Average Secrecy Capacity}

Mathematically, for Scenario-I, ASC is expressed as \cite [Eq.~(15)]{islam2020secrecy}
\begin{align}
    ASC=\int_{0}^{\infty}\frac{1}{1+\gamma}F_{\gamma_{\mathcal{E}}}(\gamma)\left [ 1-F_{\gamma_{eq}}(\gamma ) \right ]\text{d}\gamma.
    \label{eq:eq30}
\end{align}
Upon substituting \eqref{eq:eq24} and \eqref{eq:eq29} into \eqref{eq:eq30}, and then performing integration, ASC is expressed as
\begin{align}
\nonumber
ASC & =  \sum_{i_{\mathcal{E}}=1}^{m}\sum_{k_{\mathcal{E}}=0}^{p}K_{1_{\mathcal{E}}}\left( \mathfrak{R}_{1}-\sum_{i_{R}=1}^{m}\sum_{k_{R}=0}^{p}K_{1_{R}}\mathfrak{R}_{2} \right.
\\
& -\Upsilon_{3_{\mathcal{D}}}\mathfrak{R}_{3}+\left.\sum_{i_{R}=1}^{m}\sum_{k_{R}=0}^{p}K_{1_{R}}\Upsilon_{3_{\mathcal{D}}}\mathfrak{R}_{4} \right),
\label{eq:eq31}
\end{align}
where the four integral parts $\mathfrak{R}_{1}$, $\mathfrak{R}_{2}$, $\mathfrak{R}_{3}$ and $\mathfrak{R}_{4}$ are derived as follows.
\begin{remark}
The ASC expression in \eqref{eq:eq31} represents the average rate at which secure information can be transmitted over a communication channel. It is a measure of the system's overall ability to maintain security, considering variations in channel conditions. A higher average secrecy capacity indicates a more consistently secure communication system. 
\end{remark}

\subsubsection{Derivation of $\mathfrak{R}_{1}$} 
$\mathfrak{R}_{1}$ is expressed as
\begin{align}
\nonumber
    \mathfrak{R}_{1}=\int_{0}^{\infty}(1+\gamma )^{-1}\text{G}_{2,3}^{1,2} \left[K_{2_{\mathcal{E}}}\gamma  \middle| \begin{matrix}1-\Psi_{i_{\mathcal{E}}},1 \\
                                    L_{\mathcal{E}}\mu_{\mathcal{E}}+k_{\mathcal{E}},0,-\Psi_{i_{\mathcal{E}}}
                          \end{matrix}\right]\text{d}\gamma. 
                          \label{eq:eq32}
  \end{align} 
To convert $(1+\gamma)^{-1}$ into Meijer's $G$ functions and then to do integration, we use the identities from \cite [Eq.~(8.4.2.5)] {prudnikov1988integrals} and \cite [Eq.~(2.24.1.2)] {prudnikov1988integrals}. Finally, we obtain $\mathfrak{R}_{1}$ as the following.
\begin{align}
\nonumber
\mathfrak{R}_{1}&=\int_{0}^{\infty}\text{G}_{1,1}^{1,1} \left[\gamma  \middle| \begin{matrix}0 \\
                                        0
                          \end{matrix}\right]\text{G}_{2,3}^{1,2} \left[K_{2_{\mathcal{E}}}\gamma  \middle| \begin{matrix}1-\Psi_{i_{\mathcal{E}}},1 \\
                                        L_{\mathcal{E}}\mu_{\mathcal{E}}+k_{\mathcal{E}},0,-\Psi_{i_{\mathcal{E}}}
                          \end{matrix}\right]\text{d}\gamma \\
                          &=\text{G}_{3,4}^{2,3} \left[K_{2_{\mathcal{E}}} \middle| \begin{matrix}1,0 \\
                                        L_{\mathcal{E}}\mu_{\mathcal{E}}+k_{\mathcal{E}},0,0,-\Psi_{i_{\mathcal{E}}}
                          \end{matrix}\right].
\end{align}

\subsubsection{Derivation of $\mathfrak{R}_{2}$} 
$\mathfrak{R}_{2}$ is expressed as
\begin{align}
    \nonumber
      \mathfrak{R}_{2}&=\int_{0}^{\infty}(1+\gamma)^{-1}\text{G}_{2,3}^{1,2} \left[K_{2_{\mathcal{E}}}\gamma  \middle| \begin{matrix}1-\Psi_{i_{\mathcal{E}}},1 \\
                                        L_{\mathcal{E}}\mu_{\mathcal{E}}+k_{\mathcal{E}},0,-\Psi_{i_{\mathcal{E}}}
                          \end{matrix}\right]
\\\nonumber
                         & \times\text{G}_{2,3}^{1,2} \left[K_{2_{R}}\gamma  \middle| \begin{matrix}1-\Psi _{i_{R}},1 \\
                                        L_{R}\mu_{R}+k_{R},0,-\Psi _{i_{R}}
                          \end{matrix}\right]\text{d}\gamma.   
\end{align}
 Utilizing \cite [Eq.~(1.112.1)] {zwillinger2007table} and carrying out integration with the aid of \cite [Eq.~(2.24.1.1)] {prudnikov1988integrals}, $\mathfrak{R}_{2}$ is obtained as
\begin{align}
    \nonumber
        \mathfrak{R}_{2}&=\sum_{z=1}^{\infty}(-1)^{z-1}\int_{0}^{\infty}\gamma^{z-1}\text{G}_{2,3}^{1,2} \left[K_{2_\mathcal{E}}\gamma  \middle| \begin{matrix}1-\Psi_{i_{\mathcal{E}}},1 \\
                                        L_{\mathcal{E}}\mu_{\mathcal{E}}+k_{\mathcal{E}},0,-\Psi_{i_{\mathcal{E}}}
                          \end{matrix}\right]\\
                          \nonumber
                         & \times\text{G}_{2,3}^{1,2} \left[K_{2_{R}}\gamma  \middle| \begin{matrix}1-\Psi _{i_{R}},1 \\
                                        L_{R}\mu_{R}+k_{R},0,-\Psi_{i_{R}}
                          \end{matrix}\right]\text{d}\gamma\\
                          &=\sum_{z=1}^{\infty}\frac{(-1)^{z-1}}{K_{2_{\mathcal{E}}}^{z}}\text{G}_{5,5}^{3,3} \left[\frac{K_{2_{R}}}{K_{2_{\mathcal{E}}}} \middle| \begin{matrix}1-\Psi_{i_{R}},1,\Omega_{1},\Omega_{2},\Omega_{3} \\
                                        L_{R}\mu_{R}+k_{R},\Omega_{4},-z,0,-\Psi _{i_{R}}
                          \end{matrix}\right]             
    \end{align} where $\Omega_{1}=1-z-L_{\mathcal{E}}\mu_{\mathcal{E}}-k_{\mathcal{E}}$, $\Omega_{2}=1-z$, $\Omega_{3}=1-z+\Psi _{i_{\mathcal{E}}}$ and $\Omega_{4}=\Psi_{i_{\mathcal{E}}}-z$.

\subsubsection{Derivation of $\mathfrak{R}_{3}$} 

$\mathfrak{R}_{3}$ is expressed as
\begin{align}
\nonumber
         \mathfrak{R}_{3}&=\int_{0}^{\infty}(1+\gamma)^{-1}\gamma ^{\frac{\rho_{\mathcal{D}}}{r_{\mathcal{D}}}}
                          \text{G}_{2,3}^{1,2} \left[K_{2_{\mathcal{E}}}\gamma  \middle| \begin{matrix}1-\Psi_{i_{\mathcal{E}}},1 \\
                                        L_{\mathcal{E}}\mu_{\mathcal{E}}+k_{\mathcal{E}},0,-\Psi_{i_{\mathcal{E}}}
                          \end{matrix}\right]\\
                          &\times\text{G}_{1,r_{\mathcal{D}}+1}^{r_{\mathcal{D}},1} \left[\Upsilon_{4_{\mathcal{D}}}\gamma \middle| \begin{matrix} 1-\frac{\rho_{\mathcal{D}}}{r_{\mathcal{D}}} \\
                                        0,\Upsilon_{5_{\mathcal{D}}},-\frac{\rho_{\mathcal{D}}}{r_{\mathcal{D}}}
                          \end{matrix}\right]\text{d}\gamma.    
\end{align}
 Making use of \cite [Eq.~(1.112.1)] {zwillinger2007table} and implementing integration via \cite [Eq.~(2.24.1.1)] {prudnikov1988integrals}, $\mathfrak{R}_{2}$ is written as
\begin{align}
\nonumber
\mathfrak{R}_{3}&=\sum_{z=1}^{\infty}(-1)^{z-1}\int_{0}^{\infty}\gamma ^{\varpi -1}\text{G}_{2,3}^{1,2} \left[K_{2_{\mathcal{E}}}\gamma  \middle| \begin{matrix}1-\Psi_{i_{\mathcal{E}}},1 \\
L_{\mathcal{E}}\mu_{\mathcal{E}}+k_{\mathcal{E}},0,-\Psi_{i_{\mathcal{E}}}
\end{matrix}\right]\\\nonumber
&\times\text{G}_{1,r_{\mathcal{D}}+1}^{r_{\mathcal{D}},1} \left[\Upsilon_{4_{\mathcal{D}}}\gamma \middle| \begin{matrix} 1-\frac{\rho_{\mathcal{D}}}{r_{\mathcal{D}}} \\
0,\Upsilon_{5_{\mathcal{D}}},-\frac{\rho_{\mathcal{D}}}{r_{\mathcal{D}}}
\end{matrix}\right]\text{d}\gamma\\
&=\sum_{z=1}^{\infty}\frac{(-1)^{z-1}}{K_{2_{\mathcal{E}}}^{\varpi }}\text{G}_{4,r_{\mathcal{D}}+3}^{r_{\mathcal{D}}+2,2} \left[\frac{\Upsilon_{4_{\mathcal{D}}}}{K_{2_{\mathcal{E}}}} \middle| \begin{matrix}\Omega_{5} , \Omega_{6},\Omega_{7} ,\Omega _{8} \\
S_{1r_{\mathcal{D}}},\Omega_{9},-\varpi,  S_{2r_{\mathcal{D}}} 
\end{matrix}\right],
\end{align}
where $\varpi=\frac{\rho_{\mathcal{D}}}{r_{\mathcal{D}}}+z$, $\Omega_{5}=1-\frac{\rho_{\mathcal{D}}}{r_{\mathcal{D}}}$, $\Omega_{6}=1-\varpi -L_{\mathcal{E}}\mu_{\mathcal{E}}-k_{\mathcal{E}}$, $\Omega_{7}=1-\varpi$, $\Omega_{8}=1-\varpi +\Psi_{i_{\mathcal{E}}}$, $\Omega_{9}=\Psi_{i_{\mathcal{E}}}-\varpi$, $S_{11}=0$, $S_{12}=(0,\Upsilon_{5_{\mathcal{D}}})$, $S_{21}=\Upsilon_{5_{\mathcal{D}}}$ and $S_{22}=-\frac{\rho_{\mathcal{D}}}{r_{\mathcal{D}}}$. 

\subsubsection{Derivation of $\mathfrak{R}_{4}$} 
$\mathfrak{R}_{4}$ is expressed as
\begin{align}
    \nonumber
         \mathfrak{R}_{4}&=\int_{0}^{\infty}(1+\gamma)^{-1}\gamma ^{\frac{\rho_{\mathcal{D}}}{r_{\mathcal{D}}}}\text{G}_{2,3}^{1,2} \left[K_{2_{\mathcal{E}}}\gamma  \middle| \begin{matrix}1-\Psi_{i_{\mathcal{E}}},1 \\
                                        L_{\mathcal{E}}\mu_{\mathcal{E}}+k_{\mathcal{E}},0,-\Psi_{i_{\mathcal{E}}}
                          \end{matrix}\right]\\
                          \nonumber
                          &\times\text{G}_{1,r_{\mathcal{D}}+1}^{r_{\mathcal{D}},1} \left[\Upsilon_{4_{\mathcal{D}}}\gamma \middle| \begin{matrix} 1-\frac{\rho_{\mathcal{D}}}{r_{\mathcal{D}}} \\
                                        0,\Upsilon_{5_{\mathcal{D}}},-\frac{\rho_{\mathcal{D}}}{r_{\mathcal{D}}}
                          \end{matrix}\right]\\
                          &\times\text{G}_{2,3}^{1,2} \left[K_{2_{R}}\gamma  \middle| \begin{matrix}1-\Psi _{i_{R}},1 \\
                                        L_{R}\mu_{R}+k_{R},0,-\Psi _{i_{R}}
                          \end{matrix}\right]\text{d}\gamma. 
\end{align}
Applying \cite [Eq.~(1.112.1)] {zwillinger2007table} and implementing integration with the aid of \cite [Eq.~(07.34.21.0081.01)] {wolfram1999mathematica}, $\mathfrak{R}_{4}$ is expressed as
\begin{align}
\nonumber
\mathfrak{R}_{4}&=\sum_{z=1}^{\infty}(-1)^{z-1}
\int_{0}^{\infty}\gamma^{\varpi-1}\text{G}_{2,3}^{1,2} \left[K_{2_\mathcal{E}}\gamma  \middle| \begin{matrix}1-\Psi_{i_{\mathcal{E}}},1 \\
L_{\mathcal{E}}\mu_{\mathcal{E}}+k_{\mathcal{E}},0,-\Psi_{i_{\mathcal{E}}}
\end{matrix}\right]\\\nonumber
&\times\text{G}_{1,r_{\mathcal{D}}+1}^{r_{\mathcal{D}},1} \left[\Upsilon_{4_{\mathcal{D}}}\gamma \middle| \begin{matrix} 1-\frac{\rho_{\mathcal{D}}}{r_{\mathcal{D}}} \\
0,\Upsilon_{5_{\mathcal{D}}},-\frac{\rho_{\mathcal{D}}}{r_{\mathcal{D}}}
\end{matrix}\right]
\\\nonumber
&\times\text{G}_{2,3}^{1,2} \left[K_{2_{R}}\gamma  \middle| \begin{matrix}1-\Psi _{i_{R}},1 \\
L_{R}\mu_{R}+k_{R},0,-\Psi_{i_{R}}
\end{matrix}\right]\text{d}\gamma\\
&=\sum_{z=1}^{\infty}\frac{(-1)^{z-1}}{K_{2_{\mathcal{E}}}^{\varpi}}\times\\\nonumber
&\text{G}_{3,2\colon 1,r_{\mathcal{D}}+1\colon 2,3}^{2,1\colon r_{\mathcal{D}},1\colon 1,2} \left[\begin{matrix}\Omega _{6},\Omega_{7},\Omega_{8} \\
\Omega _{10},-\varpi
\end{matrix}  \middle| \begin{matrix}\Omega_{5}\\
S_{3r_{D}}
\end{matrix}\middle|\begin{matrix}1-\Psi _{i_{R}},1\\
L_{R}\mu_{R}+k_{R},0,-\Psi_{i_{R}}
\end{matrix}  \middle| \Omega_{11}
\right ],
\end{align}
where $\Omega_{10}=\Psi_{i_{\mathcal{E}}}-\varpi$, $\Omega_{11}=\left ( \frac{\Upsilon_{4_{\mathcal{D}}}}{K_{2_{\mathcal{E}}}},\frac{K_{2_{R}}}{K_{2_{\mathcal{E}}}} \right ), $ $S_{31}=(0,\Upsilon_{5_{\mathcal{D}}})$, $S_{32}=(0,\Upsilon_{5_{\mathcal{D}}},-\frac{\rho_{\mathcal{D}}}{r_{\mathcal{D}}})$, and $\text{G}_{-,-\colon -,-\colon -,-}^{-,-\colon -,-\colon -,-}\left[\cdot\mid\cdot\mid\cdot\mid\cdot\right]$ denotes the Bivariate Meiger's $G$-function \cite{chergui2016performance}.

\section{Secrecy Outage Probability}
Mathematically, the SOP can be defined as the probability that the instantaneous secrecy capacity drops below a certain threshold chosen to represent an acceptable level of security. In the following subsections, we derive the expression of SOP considering the three eavesdropping scenarios.

\subsection{Scenario-I}
The SOP is defined in the case of RF eavesdropping as
\begin{align}
\label{dkj}
\nonumber
   SOP^{I}&=\Pr(\gamma_{eq}\leq \phi+\phi\gamma_{\mathcal{E}}-1)
   \\
   &=\int_{0}^{\infty}F_{\gamma _{eq}}(\phi+\phi\gamma-1 )f_{\gamma _{\mathcal{E}}}(\gamma)d\gamma,
\end{align}
$\phi =2^{\mathcal{R}_{s}}$ and $\mathcal{R}_{s}>0$. Obtaining the exact expression for the SOP is challenging due to the involvement of shifts in Meiger's $G$ function. Consequently, researchers often resort to deriving lower bounds for SOP. Drawing on the formulation presented in \cite[Eq.~(21)]{sarker2020secrecy}, the lower bound of SOP can be expressed as
\begin{align}
\nonumber
    SOP^{I}\geq SOP_{L}^{I}&=\Pr\left \{ \gamma _{eq} \leq \phi \gamma _{\mathcal{E}}\right \}
    \\
    &=\int_{0}^{\infty}F_{\gamma _{eq}}(\phi \gamma )f_{\gamma _{\mathcal{E}}}(\gamma)\text{d}\gamma. 
    \label{eq:eq40}
\end{align}
\begin{remark}
It is noted that comparing \eqref{dkj} with \eqref{eq:eq40}, the tightness between the exact SOP and lower bound for SOP can easily be found which 
depends on $\phi=2^{\mathcal{R}_{s}}$ and $\gamma_{\mathcal{E}}$. It can be easily observed that a smaller $\mathcal{R}_{s}$ and/or a larger $\bar{\gamma}_{\mathcal{E}}$ is responsible for a tighter bound.
\end{remark}

Finally, SOP$^{I}_{L}$ is obtained after putting (\ref{eq:eq23}) and (\ref{eq:eq29}) into (\ref{eq:eq40}) and then performing integration as the following.
\begin{align}
\nonumber
SOP^{I}_{L}&=\sum_{i_{\mathcal{E}}=1}^{m}\sum_{k_{\mathcal{E}}=0}^{p}K_{1_{\mathcal{E}}}\biggl ( \sum_{i_{R}=1}^{m}\sum_{k_{R}=0}^{p}K_{1_{R}}\mathfrak{Z}_{1}
\\
&+{\Upsilon}^{'}_{3_{\mathcal{D}}}\mathfrak{Z}_{2}-\sum_{i_{R}=1}^{m}\sum_{k_{R}=0}^{p}K_{1_{R}}{\Upsilon}^{'}_{3_{\mathcal{D}}}\mathfrak{Z}_{3} \biggl ), 
     \label{eq:eq41}
\end{align}
where ${\Upsilon}^{'}_{3_{\mathcal{D}}}=\Upsilon_{3_{\mathcal{D}}}\phi^{\frac{\rho_{\mathcal{D}}}{r_{\mathcal{D}}}}$ and $\mathfrak{Z}_{1}$, $\mathfrak{Z}_{2}$ and $\mathfrak{Z}_{3}$ are three integral terms.
\begin{remark}
The derived expression in \eqref{eq:eq41} can be used to quantify the probability of a communication link failing to maintain the secrecy of the transmitted confidential information. It mainly assesses the likelihood that this secrecy is compromised or outed during the communication process. 
\end{remark}
The three integral parts $\mathfrak{Z}_{1}$, $\mathfrak{Z}_{2}$ and $\mathfrak{Z}_{3}$ are derived in the following manner.
\subsubsection*{Derivation of $\mathfrak{Z}_{1}$} 
$\mathfrak{Z}_{1}$ is derived as
\begin{align}
\nonumber
\mathfrak{Z}_{1}&=\int_{0}^{\infty}\gamma ^{-1}\text{G}_{2,3}^{1,2} \left[{K}'_{2_{R}}\gamma  \middle| \begin{matrix}1-\Psi_{i_{R}},1 \\L_{R}\mu_{R}+k_{R},0,-\Psi_{i_{R}}\end{matrix}\right]\\
&\times\text{G}_{1,2}^{1,1} \left[K_{2_{\mathcal{E}}}\gamma  \middle| \begin{matrix}1-\Psi_{i_{\mathcal{E}}} \\L_{\mathcal{E}}\mu_{\mathcal{E}}+k_{\mathcal{E}},-\Psi_{i_{\mathcal{E}}}
                          \end{matrix}\right]d\gamma. 
\end{align}
Performing the integration through \cite [Eq.~(2.24.1.1)] {prudnikov1988integrals}, $\mathfrak{Z}_{1}$ is expressed as
\begin{align}
     \mathfrak{Z}_{1}=\text{G}_{4,4}^{3,2} \left[\frac{K_{2_{\mathcal{E}}}}{{K}'_{2_{R}}}  \middle| \begin{matrix}1-\Psi_{i_{\mathcal{E}}},1-L_{R}\mu_{R}-k_{R},1,1+\Psi_{i_{R}} \\
                                        L_{\mathcal{E}}\mu_{\mathcal{E}}+k_{\mathcal{E}},\Psi_{i_{R}},0,-\Psi_{i_{\mathcal{E}}}
                          \end{matrix}\right],
\end{align}
where ${K}'_{2_{R}}=\phi K_{2_{R}}$.

\subsubsection*{Derivation of $\mathfrak{Z}_{2}$} $\mathfrak{Z}_{2}$ is expressed as
\begin{align}
\nonumber
\mathfrak{Z}_{2}&=\int_{0}^{\infty}\gamma^{\frac{\rho_{\mathcal{D}}}{r_{\mathcal{D}}}-1}\text{G}_{1,r_{\mathcal{D}}+1}^{r_{\mathcal{D}},1} \left[{\Upsilon}'_{4_{\mathcal{D}}}\gamma \middle| \begin{matrix} 1-\frac{\rho_{\mathcal{D}}}{r_{\mathcal{D}}} \\
                                        0,\Upsilon_{5_{\mathcal{D}}},-\frac{\rho_{\mathcal{D}}}{r_{\mathcal{D}}}
                          \end{matrix}\right]\\
                          &\times\text{G}_{1,2}^{1,1} \left[K_{2_{\mathcal{E}}}\gamma  \middle| \begin{matrix}1-\Psi_{i_{\mathcal{E}}} \\
                                        L_{\mathcal{E}}\mu_{\mathcal{E}}+k_{\mathcal{E}},-\Psi_{i_{\mathcal{E}}}
                          \end{matrix}\right]\text{d}\gamma.  
\end{align}
After implementing the integration identity of \cite [Eq.~(2.24.1.1)] {prudnikov1988integrals}, $\mathfrak{Z}_{2}$ is expressed as
\begin{align}
    \mathfrak{Z}_{2}=({\Upsilon}'_{4_{\mathcal{D}}})^{-\left ( \frac{\rho_{\mathcal{D}}}{r_{\mathcal{D}}}+1 \right )}\text{G}_{2+r_{\mathcal{D}},3}^{2,1+r_{\mathcal{D}}} \left[\frac{K_{2_{\mathcal{E}}}}{{\Upsilon}'_{4_{\mathcal{D}}}}  \middle| \begin{matrix}1-\Psi_{i_{\mathcal{E}}}, S_{4r_{\mathcal{D}}} \\
                                        L_{\mathcal{E}}\mu_{\mathcal{E}}+k_{\mathcal{E}},0,-\Psi_{i_{\mathcal{E}}}
                          \end{matrix}\right],
\end{align}
where ${\Upsilon}'_{4_{\mathcal{D}}}=\phi\Upsilon_{4_{\mathcal{D}}}$, $S_{41}=(1-\frac{\rho_{\mathcal{D}}}{r_{\mathcal{D}}},1-\frac{\rho_{\mathcal{D}}}{r_{\mathcal{D}}}-\Upsilon_{5_{\mathcal{D}}})$ and $S_{42}=(1-\frac{\rho_{\mathcal{D}}}{r_{\mathcal{D}}},1-\frac{\rho_{\mathcal{D}}}{r_{\mathcal{D}}}-\Upsilon_{5_{\mathcal{D}}},1)$.\\

\subsubsection*{Derivation of $\mathfrak{Z}_{3}$} $\mathfrak{Z}_{3}$ is expressed as
\begin{align}
\nonumber
\mathfrak{Z}_{3}&=\int_{0}^{\infty}\gamma^{\frac{\rho_{\mathcal{D}}}{r_{\mathcal{D}}}-1}\text{G}_{2,3}^{1,2} \left[{K}'_{2_{R}}\gamma  \middle| \begin{matrix}1-\Psi_{i_{R}},1 \\
L_{R}\mu_{R}+k_{R},0,-\Psi _{i_{R}}
\end{matrix}\right]
\\\nonumber
                          &\times\text{G}_{1,r_{\mathcal{D}}+1}^{r_{\mathcal{D}},1} \left[{\Upsilon }'_{4_{\mathcal{D}}}\gamma \middle| \begin{matrix} 1-\frac{\rho_{\mathcal{D}}}{r_{\mathcal{D}}} \\
                                        0,\Upsilon_{5_{\mathcal{D}}},-\frac{\rho_{\mathcal{D}}}{r_{\mathcal{D}}}
                          \end{matrix}\right]\\
                          &\times\text{G}_{1,2}^{1,1} \left[K_{2_{\mathcal{E}}}\gamma  \middle| \begin{matrix}1-\Psi_{i_{\mathcal{E}}} \\
                                        L_{\mathcal{E}}\mu_{\mathcal{E}}+k_{\mathcal{E}},-\Psi_{i_{\mathcal{E}}}
                          \end{matrix}\right]\text{d}\gamma.
\end{align}
$\mathfrak{Z}_{3}$ is represented as after integrating with the aid of \cite [Eq.~(07.34.21.0081.01)] {wolfram1999mathematica} as
\begin{align}
    \nonumber
         \mathfrak{Z}_{3}&=\left ( {K}'_{2_{R}} \right )^{-\frac{\rho_{\mathcal{D}}}{r_{\mathcal{D}}}}\times\\
         &\text{G}_{3,2\colon 1,r_{\mathcal{D}}+1\colon 1,2}^{2,1\colon r_{\mathcal{D}},1\colon 1,1} \left[\begin{matrix}\Omega _{12},\Omega_{13},\Omega_{14} \\
                                        \Omega _{15},-\frac{\rho_{\mathcal{D}}}{r_{\mathcal{D}}}
                          \end{matrix}  \middle| \begin{matrix}\Omega_{16}\\
                                        S_{3r_{\mathcal{D}}}
                          \end{matrix}\middle|\begin{matrix}1-\Psi_{i_{\mathcal{E}}},1\\
                                        L_{\mathcal{E}}\mu_{\mathcal{E}}+k_{\mathcal{E}},-\Psi_{i_{\mathcal{E}}}
 \end{matrix}  \middle|\Omega_{17}
\right ],
\end{align}
where $\Omega_{12}=1-\frac{\rho_{\mathcal{D}}}{r_{\mathcal{D}}}-L_{R}\mu_{R}-k_{R}$, $\Omega_{13}=1-\frac{\rho_{\mathcal{D}}}{r_{\mathcal{D}}}$, $\Omega_{14}=1-\frac{\rho_{\mathcal{D}}}{r_{\mathcal{D}}}+\Psi_{i_{R}}$, $\Omega_{15}=\Psi_{i_{R}}-\frac{\rho_{\mathcal{D}}}{r_{\mathcal{D}}}$, $\Omega_{16}=1-\frac{\rho_{\mathcal{D}}}{r_{\mathcal{D}}}$, and $\Omega_{17}=\left ( \frac{{\Upsilon}'_{4_{\mathcal{D}}}}{{K}'_{2_{R}}},\frac{K_{2_{\mathcal{E}}}}{{K}'_{2_{R}}} \right )$.

\subsubsection*{Asymptotic Analysis}
The asymptotic expression of $SOP^{I}_{L}$ can be obtained as
\begin{align}
SOP^{I}_{\infty } =\int_{0}^{\infty}F^{\infty }_{\gamma _{eq}}(\phi \gamma )f_{\gamma _{\mathcal{E}}}(\gamma)d\gamma.    \label{eq:eq48n}
\end{align}
Substituting \eqref{eq:eq23} and \eqref{eq:eq49e} into \eqref{eq:eq48n}, and performing integration through \cite [Eq.~(7.811.4)]  {zwillinger2007table}, $SOP^{I}_{\infty}$ is obtained as shown in \eqref{eq:eqasy2}.
\begin{figure*}
\setcounter{equation}{48}
    \begin{align}
        \nonumber
        SOP^{I}_{\infty }&=\sum_{i_{\mathcal{E}}=1}^{m}\sum_{k_{\mathcal{E}}=0}^{p}K_{1_{\mathcal{E}}}\left ( \sum_{i_{R}=1}^{m}\sum_{k_{R}=0}^{p}\frac{K_{1_{R}}}{(K_{2_{R}}\phi)^{-L_{R}\mu_{R}-k_{R}}}\frac{\prod_{l_{1}=1}^{2}\Gamma(T_{1,l_{1}}+L_{R}\mu_{R}+k_{R})}{\prod_{l_{1}=2}^{3}\Gamma(T_{2,l_{1}}+L_{R}\mu_{R}+k_{R})}\frac{\Gamma \left ( L_{R}\mu_{R}+k_{R}+L_{\mathcal{E}}\mu_{\mathcal{E}}+k_{\mathcal{E}}  \right )}{k_{2_{\mathcal{E}}}^{L_{R}\mu_{R}+k_{R}}\left ( \Psi_{i_{\mathcal{E}}}-L_{R}\mu_{R}-k_{R}  \right )}\right.\\
        \nonumber
        &+\left.\sum_{k_{2}=1}^{r_{\mathcal{D}}}\frac{\Upsilon_{3_{\mathcal{D}}}\phi ^{\frac{\rho_{\mathcal{D}}}{r_{\mathcal{D}}}}}{\Upsilon_{4_{\mathcal{D}}}^{T_{3,k_{2}}-1}}
          \frac{\prod_{l_{2}=1;l_{2}\neq k_{2}}^{r_{\mathcal{D}}}\Gamma(T_{3,k_{2}}-T_{3,l_{2}})\Gamma(1+\frac{\rho_{\mathcal{D}}}{r_{\mathcal{D}}}-T_{3,k_{2}})}{\prod_{l_{2}=r_{\mathcal{D}}+1}^{r_{\mathcal{D}}+1}\Gamma(1+T_{3,l_{2}}-T_{3,k_{2}})}\right.\\ 
         &\times\left.\frac{\Gamma \left ( L_{\mathcal{E}}\mu_{\mathcal{E}}+k_{\mathcal{E}}+\frac{\rho_{\mathcal{D}}}{r_{\mathcal{D}}}-T_{3,k_{2}}+1 \right )\Gamma \left ( \Psi_{i_{\mathcal{E}}} -\frac{\rho_{\mathcal{D}}}{r_{\mathcal{D}}}+T_{3,k_{2}}-1 \right )}{K_{2_{\mathcal{E}}}^{\frac{\rho_{\mathcal{D}}}{r_{\mathcal{D}}}-T_{3,k_{2}}+1}\Gamma \left ( \Psi_{i_{\mathcal{E}}} -\frac{\rho_{\mathcal{D}}}{r_{\mathcal{D}}}+T_{3,k_{2}} \right )} \right )  
        \label{eq:eqasy2}
    \end{align}
    \hrulefill
\end{figure*}

\begin{figure*}[!t]
\setcounter{equation}{51}
\begin{align}
\nonumber
     SOP_{L}^{II}&={\Upsilon}'_{3_{\mathcal{D}}}\Upsilon_{1_{\mathcal{\tilde{E}}}}\frac{r_{\mathcal{\tilde{E}}}^{\frac{1}{2}}({\Upsilon}'_{4_{\mathcal{D}}})^{-\zeta}}{(2\pi)^{\frac{1}{2}(r_{\mathcal{\tilde{E}}}-1)}}\text{G}_{r_{\mathcal{D}}+1,r_{\mathcal{\tilde{E}}}+1}^{r_{\mathcal{\tilde{E}}}+1,r_{\mathcal{D}}} \left[ \Omega_{18} \middle|\begin{matrix} S_{5r_{\mathcal{D}}}\\
                                       0,\Omega _{19}
                          \end{matrix}\right] \left ( 1-\sum_{i_{R}=1}^{m} \sum_{k_{R}=0}^{p}K_{1_{R}}\text{G}_{2,3}^{1,2} \left[\Omega_{20}  \middle| 
\begin{matrix}1-\Psi _{i_{R}},1 \\
L_{R}\mu_{R}+k_{R},0,-\Psi _{i_{R}}
\end{matrix}\right] \right )\\
& + \sum_{i_{R}=1}^{m} \sum_{k_{R}=0}^{p}K_{1_{R}}\text{G}_{2,3}^{1,2} \left[\Omega_{20}   \middle| 
\begin{matrix}1-\Psi _{i_{R}},1 \\
L_{R}\mu_{R}+k_{R},0,-\Psi _{i_{R}}
\end{matrix}\right].
\label{eq:eqasy222}
\end{align}
\hrulefill
\end{figure*}

\subsection{Scenario-II} 
Whenever the eavesdropper is present at the UOWC link, then SOP is represented as \cite[Eq.~(13)] {sarker2021intercept}
\setcounter{equation}{49}
\begin{align}
\nonumber
        SOP^{II}&=\int_{0}^{\infty} F_{\gamma_{\mathcal{D}}}(\phi \gamma +\phi -1)f_{\gamma_{\mathcal{\tilde{E}}}}(\gamma )d\gamma\\
&\times(1-F_{\gamma _{R}}(\phi -1))+F_{\gamma _{R}}(\phi -1). 
\label{eq:eq50n}
\end{align}
Since deriving the exact SOP is difficult, we define the SOP at the lower bound as
\begin{align}
\label{eq:eq51n}
\nonumber
        SOP^{II}\geq SOP_{L}^{II}
        &=\int_{0}^{\infty}F_{\gamma _{\mathcal{D}}}(\phi \gamma )f_{\gamma_{\mathcal{\tilde{E}}}}(\gamma)d\gamma\\ &\times(1-F_{\gamma _{R}}(\phi -1))+F_{\gamma _{R}}(\phi -1).  
\end{align}
Substituting \eqref{eq:eq25} and (\ref{eq:eq27}) into (\ref{eq:eq51n}) and utilizing \cite [Eq.~(2.24.1.1)] {prudnikov1988integrals}, the expression of $SOP_{L}^{II}$ is obtained as shown in \eqref{eq:eqasy222}, where ${\Upsilon}'_{3_{\mathcal{D}}}=\phi\Upsilon_{3_{\mathcal{D}}}$, $\zeta=\frac{\rho_{\mathcal{D}}}{r_{\mathcal{D}}}+\frac{\rho_{\mathcal{\tilde{E}}}}{r_{\mathcal{\tilde{E}}}}$, $\Omega_{18}=\frac{\Upsilon_{2_{\mathcal{\tilde{E}}}}^{r_{\mathcal{\tilde{E}}}}r_{\mathcal{\tilde{E}}}^{-r_{\mathcal{\tilde{E}}}}}{{\Upsilon}'_{4_{\mathcal{D}}}}$, $\Omega_{19}=-\frac{\rho_{\mathcal{\tilde{E}}}}{r_{\mathcal{\tilde{E}}}}$, $\Omega_{20}=K_{2_{R}}(\phi -1)$, $S_{51}=(1-\zeta ,1-\zeta -\Upsilon_{5_{\mathcal{D}}})$ and $S_{52}=(1-\zeta ,1-\zeta -\Upsilon _{5_{\mathcal{D}}},1-\frac{\rho_{\mathcal{D}}}{r_{\mathcal{D}}})$.
\begin{figure*}[t!]
\setcounter{equation}{52}
\begin{align}
\nonumber
    SOP^{II}_{\infty}&={\Upsilon}'_{3_{\mathcal{D}}}\Upsilon_{1_{\mathcal{\tilde{E}}}}\frac{r_{\mathcal{\tilde{E}}}^{\frac{1}{2}}({\Upsilon}'_{4_{\mathcal{D}}})^{-\zeta}}{2\pi^{\frac{1}{2}(r_{\mathcal{\tilde{E}}}-1)}}\sum_{k_{3}=1}^{r_{\mathcal{D}}} \Omega_{18}^{T_{4,k_{3}}-1}
     \frac{\prod_{l_{3}=1;l_{3}\neq k_{3}}^{r_{\mathcal{D}}}\Gamma \left ( T_{4,k_{3}}- T_{4,l_{3}}\right )\prod_{l_{3}=1}^{r_{\mathcal{\tilde{E}}}+1}\Gamma \left (1+ T_{5,l_{3}}- T_{4,k_{3}}\right )}{\prod_{l_{3}=r_{\mathcal{D}}+1}^{r_{\mathcal{D}}+1}\Gamma \left (1+ T_{4,l_{3}}- T_{4,k_{3}}\right )}\\
     \nonumber
     &\times\left ( 1-\sum_{i_{R}=1}^{m}\sum_{k_{R}=0}^{p}\frac{K_{1_{R}}}{\Omega_{20}^{-L_{R}\mu_{R}-k_{R}}}\frac{\prod_{l_{1}=1}^{2}\Gamma(T_{1,l_{1}}+L_{R}\mu_{R}+k_{R})}{\prod_{l_{1}=2}^{3}\Gamma(T_{2,l_{1}}+L_{R}\mu_{R}+k_{R})} \right )+\sum_{i_{R}=1}^{m}\sum_{k_{R}=0}^{p}\frac{K_{1_{R}}}{\Omega_{20}^{-L_{R}\mu_{R}-k_{R}}}\\
     &\times\frac{\prod_{l_{1}=1}^{2}\Gamma(T_{1,l_{1}}+L_{R}\mu_{R}+k_{R})}{\prod_{l_{1}=2}^{3}\Gamma(T_{2,l_{1}}+L_{R}\mu_{R}+k_{R})} .
\label{eq:eq53asy}    
\end{align}
\hrulefill
\end{figure*}

\subsubsection*{Asymptotic Analysis}
Upon utilizing the identity provided in \cite [Eq.~(41)] {ansari2015performance}, the asymptotic expression of $SOP_{L}^{II}$ can be obtained as shown in (\ref{eq:eq53asy}), where $T_{4}=S_{5r_{\mathcal{D}}}$ and $T_{5}=(0,\Omega_{14})$.

\subsection{Scenario-III} 
The lower bound of SOP for the RIS-aided hybrid RF-UOWC architecture with simultaneous eavesdropping attack over the RF and UOWC links is given by \cite{rahman2023ris}
\begin{align}
    SOP^{III}_{L}=1-(SOP_{1}\times SOP_{2}),
     \label{eq:eq51}
\end{align}
where
\begin{align}   
&SOP_{1}=1-\int_{0}^{\infty}F_{\gamma_{R}}(\phi\gamma)f_{\gamma_{\mathcal{E}}}(\gamma)d\gamma, \label{eq:eq52}\\
&SOP_{2}=1-\int_{0}^{\infty}F_{\gamma_{\mathcal{D}}}(\phi\gamma)f_{\gamma_{\mathcal{\tilde{E}}}}(\gamma)d\gamma. 
\label{eq:eq53}
\end{align}
Substituting (\ref{eq:eq23}) and (\ref{eq:eq24}) into (\ref{eq:eq52}) and exploiting \cite [Eq.~(2.24.1.1)] {prudnikov1988integrals}, $SOP_{1}$ is expressed as
\begin{align}
\nonumber
       SOP_{1}&=1-\sum_{i_{R}=1}^{m}\sum_{k_{R}=0}^{p}K_{1_{R}}\sum_{i_{\mathcal{E}}=1}^{m}\sum_{k_{\mathcal{E}}=0}^{p}K_{1_{\mathcal{E}}}\\
       &\times\text{G}_{4,4}^{3,2} \left[\frac{K_{2_{\mathcal{E}}}}{{K}'_{2_{R}}}  \middle| \begin{matrix}1-\Psi_{i_{\mathcal{E}}},1-L_{R}\mu_{R}-k_{R},1,1+\Psi_{i_{R}} \\
                                        L_{\mathcal{E}}\mu_{\mathcal{E}}+k_{\mathcal{E}},\Psi_{i_{R}},0,-\Psi_{i_{\mathcal{E}}}
                          \end{matrix}\right].
\end{align}
Putting (\ref{eq:eq25}) and (\ref{eq:eq27}) into (\ref{eq:eq53}) and completing the integration via \cite [Eq.~(2.24.1.1)] {prudnikov1988integrals}, $SOP_{2}$ is given by

\begin{align}
        SOP_{2}&=1-{\Upsilon}'_{3_{\mathcal{D}}}\Upsilon_{1_{\mathcal{\tilde{E}}}}\frac{r_{\mathcal{\tilde{E}}}^{\frac{1}{2}}({\Upsilon}'_{4_{\mathcal{D}}})^{-\zeta}}{(2\pi)^{\frac{1}{2}(r_{\mathcal{\tilde{E}}}-1)}}\text{G}_{r_{\mathcal{D}}+1,r_{\mathcal{\tilde{E}}}+1}^{r_{\mathcal{\tilde{E}}}+1,r_{\mathcal{D}}} \left[ \Omega_{18} \middle|\begin{matrix} S_{5r_{\mathcal{D}}}\\
                                       0,\Omega _{19}
                          \end{matrix}\right]  .
\end{align}

\subsubsection*{Asymptotic Analysis}
The asymptotic expression of $SOP^{III}_{L}$ can be written as
\begin{align}
    SOP^{\infty}_{III}=1-(SOP^{\infty}_{1}\times SOP^{\infty}_{2}).
\end{align}
Applying \cite [Eq.~(41)] {ansari2015performance}, $SOP^{\infty}_{1}$ can be expressed as shown in (\ref{eq:eqasy31}), and $SOP^{\infty}_{2}$ can be obtained as shown in (\ref{{eq:eqasy32}}), where $T_{6}=(1-\Psi_{i_{\mathcal{E}}},1-L_{R}\mu_{R}-k_{R})$ and $T_{7}=( L_{\mathcal{E}}\mu_{\mathcal{E}}+k_{\mathcal{E}},\Psi_{i_{\mathcal{E}}},0)$.

\begin{figure*}[!t]
\normalsize
\setcounter{mytempeqncnt}{\value{equation}}
\setcounter{equation}{58}
\begin{align}
    &SOP^{\infty}_{1}=1-\sum_{i_{R}=1}^{m}\sum_{k_{R}=0}^{p}K_{1_{R}}\sum_{i_{\mathcal{E}}=1}^{m}\sum_{k_{\mathcal{E}}=0}^{p}K_{1_{\mathcal{E}}}\sum_{k_{4}=1}^{2}\left ( \frac{K_{2_{\mathcal{E}}}}{{K}'_{2_{R}}} \right )^{T_{6,k_{4}}-1}
         \frac{\prod_{l_{4}=1;l_{4}\neq k_{4}}^{2}\Gamma \left ( T_{6,k_{4}}- T_{6,l_{4}}\right )\prod_{l_{4}=1}^{3}\Gamma \left (1+ T_{7,l_{4}}- T_{6,k_{4}}\right )}{\prod_{l_{4}=3}^{4}\Gamma \left (1+ T_{6,l_{4}}- T_{6,k_{4}}\right )\Gamma \left (T_{6,k_{4}}+\Psi_{i_{\mathcal{E}}}\right )}\label{eq:eqasy31}\\
    &SOP^{\infty}_{2}=1-{\Upsilon}'_{3_{\mathcal{D}}}\Upsilon_{1_{\mathcal{\tilde{E}}}}\frac{r_{\mathcal{\tilde{E}}}^{\frac{1}{2}}({\Upsilon}'_{4_{\mathcal{D}}})^{-\zeta}}{2\pi^{\frac{1}{2}(r_{\mathcal{\tilde{E}}}-1)}}\sum_{k_{3}=1}^{r_{\mathcal{D}}}\Omega_{18}^{T_{4,k_{3}}-1}
     \frac{\prod_{l_{3}=1;l_{3}\neq k_{3}}^{r_{\mathcal{D}}}\Gamma \left ( T_{4,k_{3}}- T_{4,l_{3}}\right )\prod_{l_{3}=1}^{r_{\mathcal{\tilde{E}}}+1}\Gamma \left (1+ T_{5,l_{3}}- T_{4,k_{3}}\right )}{\prod_{l_{3}=r_{\mathcal{D}}+1}^{r_{\mathcal{D}}+1}\Gamma \left (1+ T_{4,l_{3}}- T_{4,k_{3}}\right )}
     \label{{eq:eqasy32}}
\end{align}
\setcounter{equation}{\value{mytempeqncnt}}
\hrulefill
\vspace*{4pt}
\end{figure*}

\section{Strictly Positive Secrecy Capacity}

The probability of SPSC can be mathematically expressed in terms of the SOP expressions in (\ref{eq:eq41}), (\ref{eq:eqasy222}), and (\ref{eq:eq51}) as
\begin{subequations}
\begin{align}
\label{sf1}
    SPSC^{I}&=1-SOP^{I}\big|_{\mathcal{R}_{s}=0},
\\\label{sf2}
    SPSC^{II}&=1-SOP^{II}\big|_{\mathcal{R}_{s}=0},
\\\label{sf3}
    SPSC^{III}&=1-SOP^{III}\big|_{\mathcal{R}_{s}=0}.
\end{align}
\end{subequations}
\begin{remark}
If the secrecy capacity is strictly positive, it means that there is a positive rate at which covert communication can be maintained securely. The expressions in \eqref{sf1},  \eqref{sf2}, and \eqref{sf3} can be used to indicate the probability of the existence of such secrecy capacity in scenarios, where it is difficult to achieve perfect secrecy.
\end{remark}

\section{ Effective Secrecy Throughput}
When the source transmits information at a constant rate $\mathcal{R}_{s}$, then EST can be defined for the three considered scenarios as \cite {lei2020secure}
\begin{subequations}
\begin{align}
\label{kd1}
    EST^{I}&=\mathcal{R}_{s}(1-SOP^{I}),
\\\label{kd2}
    EST^{II}&=\mathcal{R}_{s}(1-SOP^{II}),
\\\label{kd3}
    EST^{III}&=\mathcal{R}_{s}(1-SOP^{III}).
\end{align}    
\end{subequations}
\begin{remark}
EST is a parameter employed in the field of information-theoretic security to measure the rate at which secret information may be transmitted over a channel while preserving a particular level of security. It is evident from \eqref{kd1}, \eqref{kd2}, and \eqref{kd3} that an increase in $\mathcal{R}_{s}$ results in a proportional increase in the SOP, and conversely, a decrease in $\mathcal{R}_{s}$ leads to a reduced SOP. This implies the existence of an optimal value for $\mathcal{R}_{s}$ that maximizes the EST under the given circumstances. Making use of \eqref{kd1}, \eqref{kd2}, and \eqref{kd3}, one can easily obtain the optimal $\mathcal{R}_{s}$ and maximum EST.
\end{remark}


\section{Numerical Results}

This section investigates the impact of various parameters on the security performance of an RIS-assisted UAV-NTN communication within the RF NTN-UOWC system. The parameters under consideration encompass fading parameters, the number of reflecting elements, RF link distance, thermal gradients, detection methods, UWT, pointing errors, and more. The objective is to comprehensively understand the intricate interplay between these factors and the overall security of the communication system. The examination involves figurative illustrations of numerical instances making use of the derived expressions, as represented in  \eqref{eq:eq31}, \eqref{eq:eq41}, \eqref{eq:eqasy222}, \eqref{eq:eq51}, \eqref{sf1}-\eqref{sf3}, and \eqref{kd2} which offer a visual representation, aiding in the interpretation of the numerical impact of each parameter on the security performance. To validate the accuracy of the newly derived expressions, MC simulations are conducted using MATLAB. In our general case assumptions, specific values are assigned to the parameters: $\alpha_{j}=2$, $\mu_{j}=2$, $\kappa_{j}=1$, $D_{j}=50$m, $L_{j}\geq1$,  $\mathcal{N}=2$, $\bar{\gamma}_{j}\geq 0$, $\bar{\gamma}_{n}\geq 0$, $r_{n}= {1,2}$, $\mathcal{R}_{s}=0.01$, and $\xi_{n}=1$. Furthermore, the UWT parameters crucial to our analysis are extracted from Table~3 in \cite{badrudduza2021security}.

\subsection{Impacts of UAV-NTN Link Parameters}
\begin{figure}[t]
\centerline{\includegraphics[width=0.37\textwidth,angle=0]{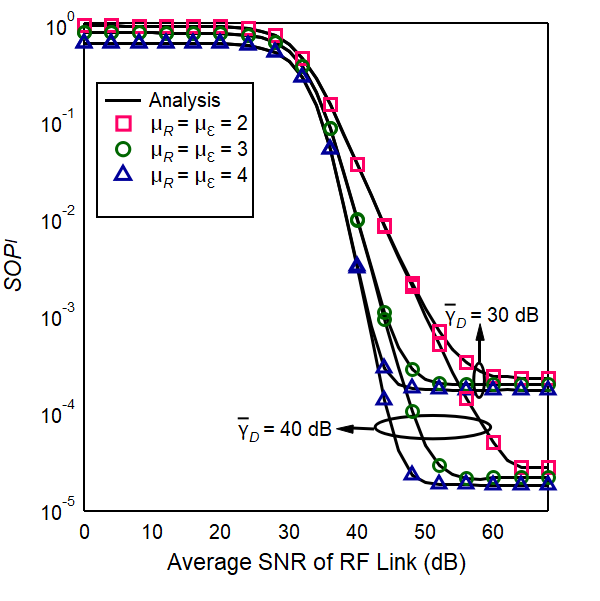}}
    \caption{The $SOP^{I}$ versus $\bar\gamma_{R}$ for selected values of $\mu_{R}$, $\mu_{\mathcal{E}}$ and $\bar\gamma_{\mathcal{D}}$.
}
    \label{fig:fig2}
\end{figure}
Figure \ref{fig:fig2} depicts the impact of the fading severity parameter on SOP by graphing the $SOP^{I}$ versus the average SNR of the main UAV-NTN link.  A notable observation from the figure emerges, demonstrating that as both $\mu_{R}$ and $\mu_{\mathcal{E}}$ are concurrently increased, there is a corresponding decrease in the secrecy outage. This reduction stems from a more challenging environment for an eavesdropper to intercept the transmitted signal effectively due to the diminished channel fading resulting from the augmentation of both $\mu_ {R}$ and $\mu_{\mathcal{E}}$, indicating increased resilience of the communication system against potential security breaches.

\begin{figure}[t]
\centerline{\includegraphics[width=0.37\textwidth,angle=0]{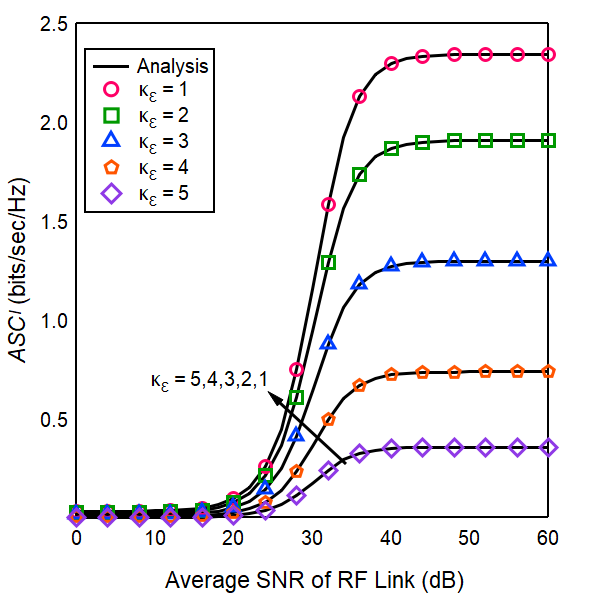}}
    \caption{The $ASC^{I}$ versus $\bar\gamma_{R}$ for selected values of $\kappa_{\mathcal{E}}$.
}
    \label{fig:fig4}
\end{figure}
The influence of small-scale fading effects due to multipath propagation on the $ASC^{I}$ is depicted in Fig. \ref{fig:fig4}. It has been observed that the ASC experiences a decline when $\kappa_{\mathcal{E}}$ is increased from $1$ to $5$. This outcome is expected, as the total fading of the eavesdropper channel diminishes with the escalation of $\kappa_{\mathcal{E}}$ towards infinity, leading to a reduction in ASC.

\begin{figure}[t]
\centerline{\includegraphics[width=0.37\textwidth,angle=0] {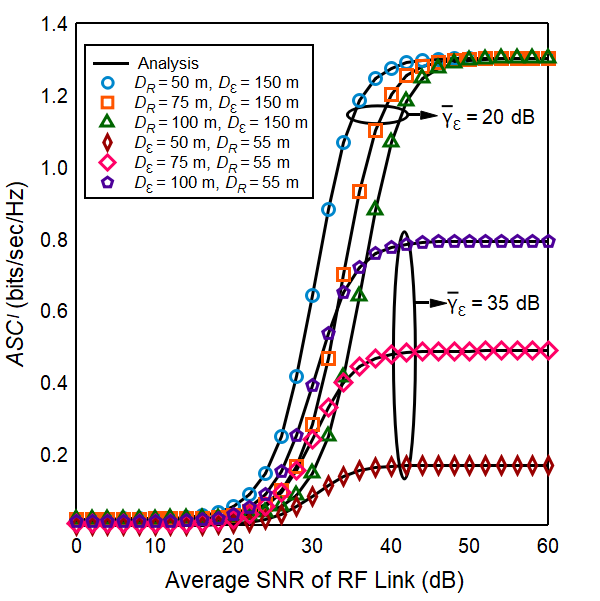}}
\caption{The $ASC^{I}$ versus $\bar\gamma_{R}$ for selected values of $D_{R}$, $D_{\mathcal{E}}$ and $\bar\gamma_{\mathcal{E}}$.
}
\label{fig:fig5}
\end{figure}
The impact of link distance in wireless communication is a pivotal factor, as demonstrated in Figure \ref{fig:fig5}, where $ASC^{I}$ is plotted against varying $D_{\mathcal{R}}$ and $D_{\mathcal{E}}$. When $D_{\mathcal{E}}$ is maintained at a constant value, the ASC is observed to decrease with an increase in $D_{\mathcal{R}}$. Conversely, when $D_{\mathcal{R}}$ is fixed, the ASC exhibits an increase with a rise in $D_{\mathcal{E}}$. This behavior is expected, considering that the signal strength at the legitimate receiver ($\mathcal{R}$) diminishes with an increase in $D_{\mathcal{R}}$ due to factors such as path loss, fading, diffraction, and scattering in free-space propagation. Moreover, heightened path loss in the eavesdropper link reduces the probability of successful eavesdropping, thereby contributing to an enhancement in secrecy. In addition to the link distance, the influence of $\bar{\gamma}_{\mathcal{E}}$ is also examined in the figure, revealing a noteworthy trend of decreased ASC with an increased value of $\bar{\gamma}_{\mathcal{E}}$. This phenomenon arises due to the higher values of $\bar{\gamma}_{\mathcal{E}}$ indicating an improved wiretap link, signifying enhanced eavesdropping capability for $\mathcal{E}$.

The impact of diversity on secrecy capacity is demonstrably positive, as illustrated in Figure~\ref{fig:fig66} by plotting $SOP^{I}$ against the UOWC link's average SNR. The plot shows that, compared to a single-branch receiver ($L_{R}=L_{\mathcal{E}}=1$), incorporating additional branches ($L_{R}=L_{\mathcal{E}} \geq 1$) at both $\mathcal{R}$ and $\mathcal{E}$ significantly improves the overall signal reception quality at $\mathcal{R}$ by reducing the probability of an outage. This improvement is due to the diversity reception technique, which effectively mitigates the fading effects and enhances the robustness of the main link. Moreover, it is evident from the figure that the impact of simultaneously increasing the number of antennas at both $\mathcal{R}$ and $\mathcal{E}$ is more advantageous for the main channel than for the eavesdropper channel. This indicates that the diversity reception technique is more critical for the main link, as it provides a more substantial improvement in signal quality and secrecy capacity compared to the eavesdropper link. Thus, the addition of multiple branches at the receiver end not only enhances the main link's performance but also contributes to a higher level of security by lowering the secrecy outage probability.

\subsection{Impacts of RIS-aided UOWC Link Parameters}

\begin{figure}[t]
\centerline{\includegraphics[width=0.37\textwidth,angle=0]{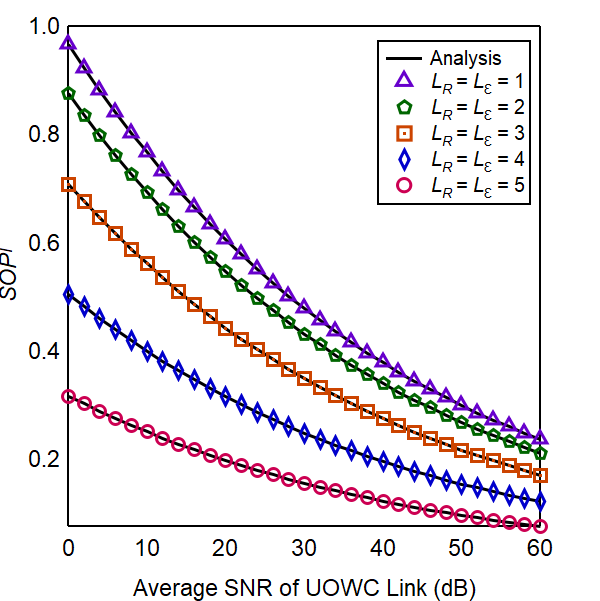}}
    \caption{The $SOP^{I}$ versus $\bar\gamma_{\mathcal{D}}$ for selected values of $L_{R}$ and $L_{\mathcal{E}}$.}
    \label{fig:fig66}
\end{figure}

\begin{figure}[t]
\centerline{\includegraphics[width=0.37\textwidth,angle=0]{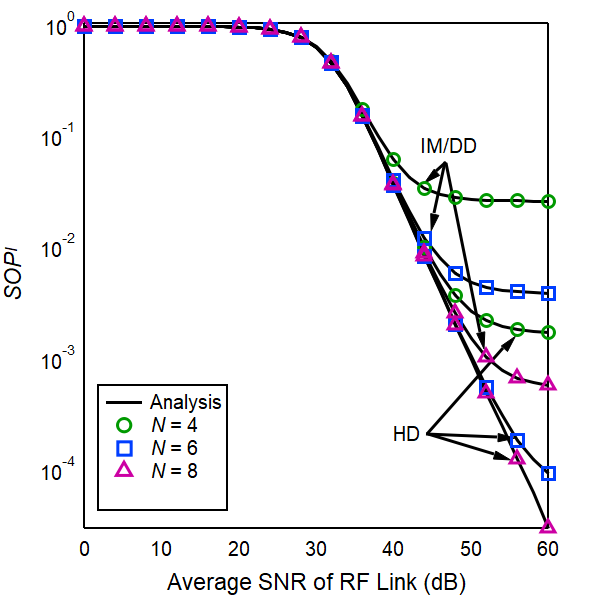}}
    \caption{The $SOP^{I}$ versus $\bar\gamma_{R}$ for selected values of $\mathcal{N}$ and $r_{n}$.}
    \label{fig:fig3}
\end{figure}

The effect of the number of reflecting elements is demonstrated in Fig. \ref{fig:fig3} in terms of $SOP^{I}$. It is evident that the SOP experiences a significant decrease with an increase in $\mathcal{N}$, as corroborated by \cite{rahman2023ris}. This reduction is attributed to the enhanced capabilities provided by a larger value of $\mathcal{N},$ allowing for more sophisticated beamforming and directional control of reflected signals. The increased $\mathcal{N}$ provides opportunities to exploit diversity and multipath effects in the wireless channel, resulting in a higher SNR at $\mathcal{D}$. Additionally, adjusting the phase shifts and amplitudes of individual elements enables the RIS to steer the transmitted signal in a specific direction, thereby reducing the probability of signal leakage to the intruders. Figure \ref{fig:fig3} also depicts a comparative analysis between HD and IM/DD techniques in terms of the SOP. Remarkably, the HD method outperforms the IM/DD method in achieving lower outage behavior more effectively, resulting in a significant enhancement in secrecy performance. The superiority of the HD method can be attributed to its integration of a local oscillator, facilitating the conversion of optical frequency to an intermediate frequency for ease of detection. This integration results in a higher SNR at the receiver compared to the IM/DD method. Additionally, the HD technique generally offers higher sensitivity due to its capability to recover both amplitude and phase information for which it is always favoured for high-speed and long distance communication.

\begin{figure}[t]
\centerline{\includegraphics[width=0.37\textwidth,angle=0]{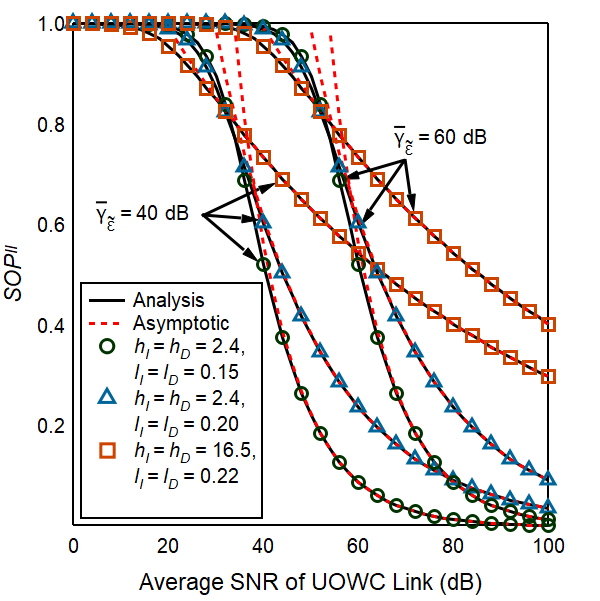}}
    \caption{The $SOP^{II}$ versus $\bar\gamma_{\mathcal{D}}$ for selected values of $h_{I}$, $h_{\mathcal{D}}$ and $\bar\gamma_{\mathcal{\tilde{E}}}$. 
    }
    \label{fig:fig6}
\end{figure}

\begin{figure}[t]
\centerline{\includegraphics[width=0.37\textwidth,angle=0]{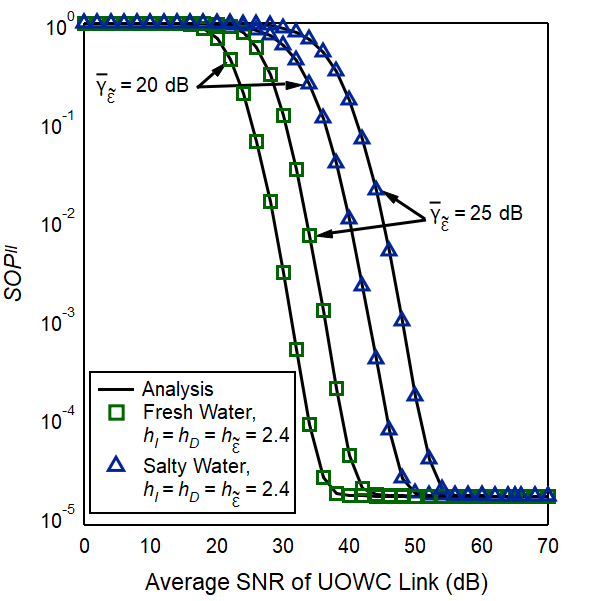}}
    \caption{The $SOP^{II}$ versus $\bar\gamma_{\mathcal{D}}$ for selected values of $h_{I}$, $h_{\mathcal{D}}$, $h_{\mathcal{\tilde{E}}}$ and $\bar\gamma_{\mathcal{\tilde{E}}}$.
}
    \label{fig:fig7}
\end{figure}
The investigation, centering on analyzing the impact of various UWT scenarios within both a thermal gradient and a thermally uniform system, is carried out by plotting the $SOP^{II}$ against $\gamma_{\mathcal{D}}$ in Figs. \ref{fig:fig6} and \ref{fig:fig7}. In Fig. \ref{fig:fig6}, a fixed $h_{\mathcal{\tilde{E}}}$ and $l_{\mathcal{\tilde{E}}}$, valued at 2.4 (L/min) and 0.10 (C $cm^{-1}$) respectively, are assumed. Notably, as the temperature gradient increases from 0.15 to 0.22, the SOP exhibits an increase, leading to a degradation in system performance. This is attributed to the amplification of irradiance variations with an increase in the temperature gradient \cite{badrudduza2021security}. The incorporation of asymptotic analysis provides further insights, and remarkably, the outcomes from simulations, asymptotic approaches, and analytical methods closely align, especially at high SNR. In Fig. \ref{fig:fig7}, the impact of UWT in two different types of water, fresh and salty, is demonstrated. It is evident that the increased salinity of water results in a higher SOP, thereby deteriorating channel performance. The findings indicate that freshwater environments are more resilient during outages compared to saltwater environments.

\begin{figure}[t]
\centerline{\includegraphics[width=0.37\textwidth,angle=0]{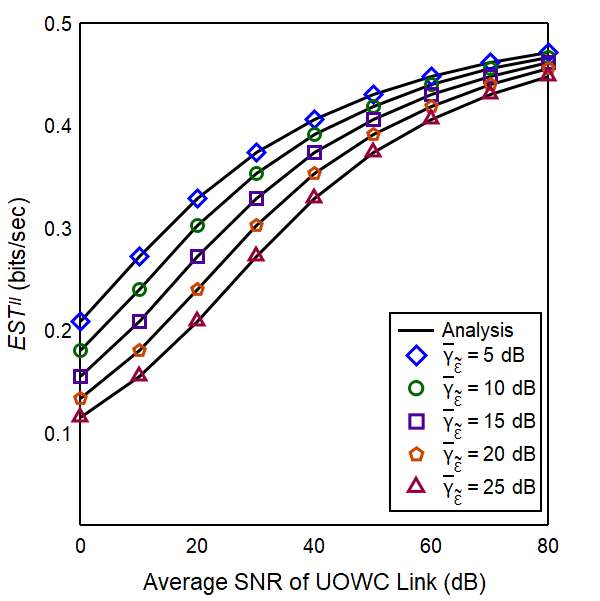}}
    \caption{The $EST^{II}$ versus $\bar\gamma_{\mathcal{D}}$ for selected values of $\bar\gamma_{\mathcal{\tilde{E}}}$. 
}
    \label{fig:fig8}
\end{figure}
The impact of eavesdropper's average SNR on $EST^{II}$ against the average SNR of the UOWC link is illustrated in Fig. \ref{fig:fig8}. As observed, the EST experiences a reduction when the average SNR of the eavesdropper increases from 5 dB to 25 dB. This outcome indicates that higher values of the eavesdropper's average SNR enhance the corresponding link's overhearing capability, subsequently declining the overall secrecy performance.

\begin{figure}[t]
\centerline{\includegraphics[width=0.37\textwidth,angle=0]{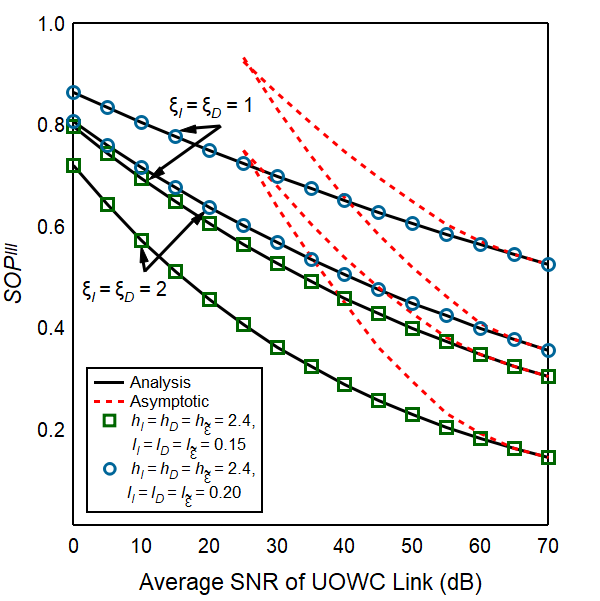}}
    \caption{The $SOP^{III}$ versus $\bar\gamma_{\mathcal{D}}$ for selected values of $h_{\mathcal{R}}$, $h_{\mathcal{D}}$, $h_{\mathcal{\tilde{E}}}$,     $\xi_{R}$ and $\xi_{\mathcal{D}}$.
}
    \label{fig:fig9}
\end{figure}
In optical wireless communication systems, pointing error refers to the deviation between the intended target and the actual direction in which a system is pointing. Figure \ref{fig:fig9} illustrates the $SOP^{III}$ against the average SNR of the UOWC link for a scenario involving simultaneous eavesdropping, showcasing the impacts of pointing errors at $D$ and $\mathcal{\tilde{E}}$. When $\xi_{\mathcal{\tilde{E}}}$ is held constant at 2, it is evident that the SOP decreases as the value of $\xi_{\mathcal{D}}$ increases from 1 to 2, thereby enhancing system performance. Similar results were also shown in \cite{ahmed2023enhancing}. Conversely, when $\xi_{\mathcal{D}}$ is fixed at 2, the SOP decreases with the growing value of $\xi_{\mathcal{\tilde{E}}}$ from 1 to 2, indicating declining misalignment error between the transmitter and the eavesdropper, thereby degrading the system performance. The asymptotic SOP is presented alongside analytical and simulation results, revealing a close alignment, particularly in the high SNR regime. This observation signifies a robust agreement with the analytical findings, emphasizing the reliability and accuracy of the asymptotic analysis in capturing system behavior, especially under conditions of high SNR.

\subsection{Comparative Analysis}
\begin{figure}[t]
\centerline{\includegraphics[width=0.37\textwidth,angle=0]{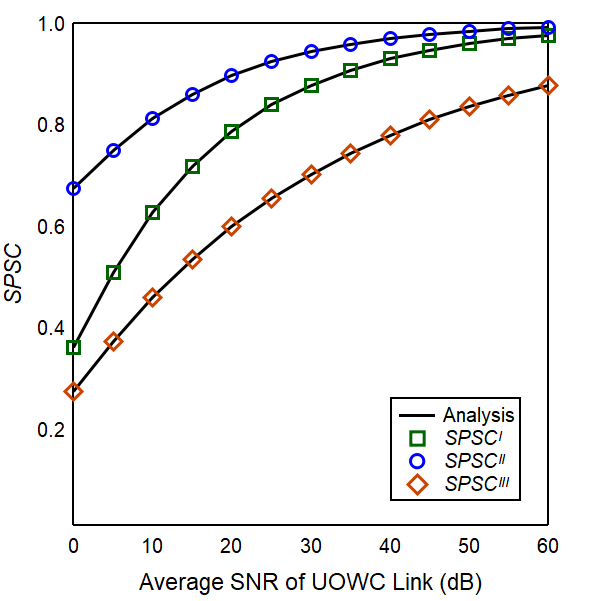}}
    \caption{The $SPSC$ versus $\bar\gamma_{\mathcal{D}}$ for three different scenario. 
}
    \label{fig:fig10}
\end{figure}

\begin{figure}[t]
\centerline{\includegraphics[width=0.37\textwidth,angle=0]{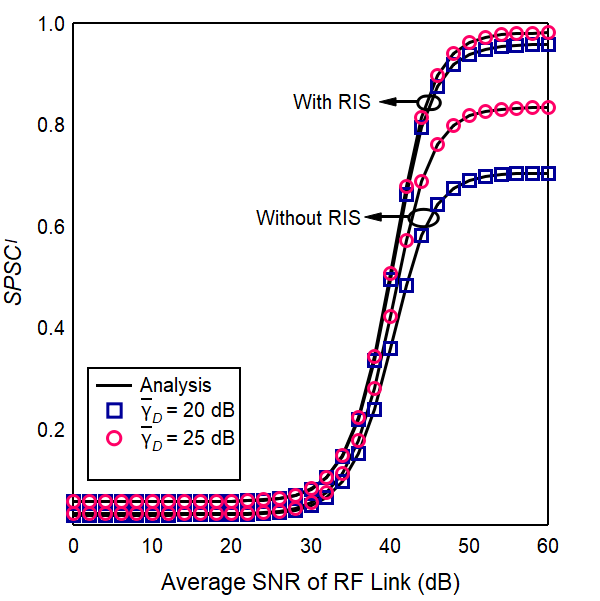}}
    \caption{The $SPSC^{I}$ versus $\bar\gamma_{\mathcal{R}}$ for selected values of $\bar\gamma_{\mathcal{D}}$.  
}
    \label{fig:fig11}
\end{figure}
The SPSC is plotted against $\bar{\gamma}_{\mathcal{D}}$ in Figure \ref{fig:fig10} and $\bar{\gamma}_{\mathcal{R}}$ in Figure \ref{fig:fig11}. A comparative analysis among three eavesdropping scenarios is presented in Figure \ref{fig:fig10}, highlighting that Scenario III exhibits the poorest secrecy performance, suggesting that simultaneous eavesdropping poses the most significant threat. Additionally, it is observed that the RF eavesdropper is more detrimental than the FSO eavesdropper, implying that the FSO link is more susceptible to eavesdropping than the RF link. Figure \ref{fig:fig11} presents a comparison between Scenario I and the same scenario without the inclusion of the RIS, revealing that better secrecy performance can be achieved when communication occurs with the aid of the RIS.

\subsection{Design Guideline}
Based on the results of the investigation into the secrecy performance of the proposed system, here are some design guidelines for consideration:

\subsubsection{RF UAV-NTN Link Modeling} 
The security of the proposed system can be ensured by choosing higher values of fading severity parameters, understanding small-scale fading effects to make a balance between the total fading and the multipath exploitation, and finally strategically selecting optimal link distances. 

\subsubsection{UOWC Link Modeling} 
In the case of the UOWC system, an increase in the number of reflecting elements in the RIS for sophisticated beamforming, considering the HD method for improved secrecy, and mitigating UWT effects to ensure robust performance in different underwater environments, considering factors such as thermal gradients and water salinity can significantly enhance secrecy. Besides, the smaller pointing errors must be addressed to obtain a better alignment between the transmitter and receiver, and the scenarios with better secrecy performance must be prioritized in the system design.

These recommendations, which address various system parameters for three considered scenarios to obtain a balance between the complexities faced by the system, utilization of limited resources, and overall security will certainly help the design engineers optimize the secrecy performance of the proposed system.

\section{Conclusions}
This study addresses a RF NTN-UOWC model that considers dynamic sources, e.g., UAVs,  and incorporates an RIS in UOWC environments. By meticulously examining three potential eavesdropping scenarios, e.g., RF NTN, UOWC, and simultaneous, we have provided valuable insights into the vulnerabilities and mitigation strategies of such networks in terms of closed-form expressions for some secrecy metrics, i.e., ASC, SOP, EST, and SPSC. Through numerical analyses and simulations, our findings highlight the significance of strategic parameter manipulation in enhancing system security, considering factors such as fading severity, number of reflecting elements, UWT, pointing errors, and detection techniques.  It is observed that by exploiting an RIS technology, we can manipulate signal propagation and minimize signal leakage to potential eavesdroppers, thereby enhancing the security of the communication system. A comparison between RIS-enabled and non-RIS systems provides remarkable insights into the efficacy of an RIS deployment in bolstering security. Furthermore, the study offers valuable design guidelines for the implementation of secure communication technologies, which can pave the way for safer and more reliable wireless communication systems in the future.

\appendices

\section{Proof of Lemma 1}
\label{A1}
In order to obtain the PDF of $\gamma_{j}$, we have to first find the PDF of $Z_{jl}=\left | h_{jl} \right |^{2} q_{j}^{-\alpha_{j}}$. Let's assume that $W_{jl}=\left | h_{jl} \right |^{2}$ and $Y_{j}=q_{j}^{-\alpha_{j}}$. Utilizing the transformation of RV, the PDF of $W_{jl}$ and $Y_{j}$ can be determined, respectively as follows.
\begin{align}
    f_{W_{jl}}(w)&=\frac{A_{j}}{2}(\sqrt{w})^{\mu _{j}-1}e^{-B_{j}w}I_{\mu_{j}-1}(M_{j}\sqrt{w})
    \label{eq:eq16}
\\
    f_{Y_{j}}(y)&=\frac{1}{\alpha _{j}}\sum_{i_{j}=1}^{m}\frac{C_{i_{j}}}{D_{j}^{\beta_{i_{j}}+1}}y^{-\frac{1}{\alpha _{j}}(\beta_{i_{j}}+1)-1},\quad D_{j}^{-\alpha_{j}}\leq q\leq \infty 
    \label{eq:eq17}
\end{align}
Hence, the PDF of $Z_{jl}$ can be obtained as
\begin{align}
    f_{Z_{jl}}(z)=\int_{D_{j}^{-\alpha_{j}}}^{\infty}\frac{1}{y}f_{W_{jl}}\left ( \frac{z}{y}\right )f_{Y_{j}}(y)dy.
      \label{eq:eq18}
\end{align}
Putting (\ref{eq:eq16}) and (\ref{eq:eq17}) to (\ref{eq:eq18}) results in
\begin{align}
\nonumber
f_{Z_{jl}}(z)&=\frac{A_{j}z^{\frac{1}{2}(\mu_{j}-1)}}{2\alpha_{j}}\sum_{i_{j}=1}^{m}\frac{C_{i_{j}}}{D_{j}^{\beta_{i_{j}}+1}}
\\\nonumber
&\times \int_{D_{j}^{-\alpha_{j}}}^{\infty} y^{-\frac{1}{2}(\mu_{j}-1)-\frac{1}{\alpha_{j}}(\beta_{i_{j}}+1)-2} 
\\       
&\times e^{-B_{j}zy^{-1}} I_{\mu_{j}-1}\left ( M_{j}\sqrt{z}y^{-\frac{1}{2}} \right ).
\label{eq:eq19}
\end{align}
Using the approximation of the Bessel function 
$I_{v_{j}}(x)=\sum_{k_{j}=0}^{p}V_{j}(k_{j},p,v_{j})\left ( \frac{x}{2} \right )^{v_{j}+2k_{j}}$, where $V_{j}(k_{j},p,v_{j})=\frac{\Gamma (p+k_{j})p^{1-2k_{j}}}{\Gamma (k_{j}+1)\Gamma(p-k_{j}+1)\Gamma (v_{j}+k_{j}+1)}
$ \cite{li2006new}, and assuming $u=y^{-1}$, (\ref{eq:eq19}) can be translated as follow. 
\begin{align}
    \nonumber
        f_{Z_{jl}}(z)&=\frac{A_{j}}{2\alpha_{j}}\sum_{i_{j}=1}^{m}\sum_{k_{j}=0}^{p}\frac{C_{i_{j}}}{D_{j}^{\beta_{i_{j}}+1}}V_{j}(k_{j},p,\mu _{j}-1)z^{\mu_{j}+k_{j}-1}
        \\
      &\times\left ( \frac{M_{j}}{2} \right )^{\mu_{j}+2k_{j}-1}\int_{0}^{D_{j}^{\alpha_{j}}}u^{\mu_{j}+k_{j}+\frac{1}{\alpha _{j}}(\beta_{i_{j}}+1)-1}e^{-B_{j}zu}du.
\end{align}
Now, the PDF of $Z_{j}$ can be deduced by utilizing \cite  [Eq.~(3.381.1)]{zwillinger2007table} as
\begin{align}
\nonumber
    f_{Z_{jl}}(z)&=\frac{A_{j}}{2\alpha_{j}}\sum_{i_{j}=1}^{m}\sum_{k_{j}=0}^{p}\frac{C_{i_{j}}}{D_{j}^{\beta_{i_{j}}+1}}V_{j}(k_{j},p,\mu _{j}-1)\\
    \nonumber
    &\times\left ( \frac{M_{j}}{2} \right )^{\mu_{j}+2k_{j}-1}
     z^{-\frac{1}{\alpha _{j}}(\beta_{i_{j}}+1)-1}B_{j}^{-\mu_{j}-k_{j}-\frac{1}{\alpha _{j}}(\beta_{i_{j}}+1)} 
    \\
    &\times \gamma\left ( \mu_{j}+k_{j}+\frac{1}{\alpha _{j}}(\beta_{i_{j}}+1), B_{j}zD_{j}^{\alpha_{j}}\right ),
    \label{eq:eq22}
\end{align} where $\gamma\left(.,.\right)$ indicates the lower incomplete Gamma function \cite [Eq.~(8.350.1)]{zwillinger2007table}.  For further simplification, employing \cite [Eq.~(8.4.16.1)] {prudnikov1988integrals} and \cite [Eq.~(8.2.2.15)] {prudnikov1988integrals} to (\ref{eq:eq22}) and following the transformation of RVs considering $\mu_{MRC_{j}}=L_{j}\mu_{j}$, $\kappa_{MRC_{j}}=\kappa_{j}$, $\bar\gamma_{MRC_{j}}=L_{j}\bar\gamma_{j}$ as stated in \cite{milisic2008outage}, the PDF of $\gamma_{j}$ can be expressed as shown in \eqref{eq:eq23}, where $\text{G}_{p,q}^{m,n} \left [ . \right ]$ indicates the Meiger's $G$-function \cite{zwillinger2007table}.

\section{Proof of Lemma 2}
\label{A2}
$F_{\gamma_{j}}(\gamma)$ can be determined through the integration of \eqref{eq:eq23} over the range [0, $\gamma_{j}$], employing the expression provided in \cite[Eq.~(2.24.2.2)]{prudnikov1988integrals} as illustrated in \eqref{eq:eq24}.

\section{Proof of Lemma 3}
\label{A3}
Making use of the identity of \cite [Eq.~(8.2.2.14)] {prudnikov1988integrals} and applying \cite [Eq.~(41)] {ansari2015performance} for expanding the Meijer's $G$ function, the asymptotic CDF of $\gamma_{j}$ can be derived as shown in \eqref{eq:eq18n}.

\bibliographystyle{IEEEtran}
\bibliography{IEEEabrv,main.bib}

\end{document}